\documentclass{amsart}%
\usepackage{amsmath}
\usepackage{amssymb}
\usepackage{hyperref}
\hypersetup{colorlinks=true}
\usepackage{amsfonts}
\usepackage{color}
\usepackage{graphicx}
\usepackage{dcolumn}
\usepackage{amsfonts}
\newcommand{\caA}{{\mathcal A}}

\newcommand{\caD}{{\mathcal D}}
\newcommand{\caE}{{\mathcal E}}

\newcommand{\caG}{{\mathcal G}}
\newcommand{\caH}{{\mathcal H}}
\newcommand{\caI}{{\mathcal I}}
\newcommand{\caJ}{{\mathcal J}}
\newcommand{\caK}{{\mathcal K}}

\newcommand{\caP}{{\mathcal P}}

\newcommand{\caS}{{\mathcal S}}
\newcommand{\caT}{{\mathcal T}}

\newcommand{\bbA}{{\mathbb A}}
\newcommand{\bbB}{{\mathbb B}}
\newcommand{\bbC}{{\mathbb C}}

\newcommand{\bbE}{{\mathbb E}}

\newcommand{\bbN}{{\mathbb N}}

\newcommand{\bbR}{{\mathbb R}}

\newcommand{\bbZ}{{\mathbb Z}}
\newcommand{\str}{^*}

\begin{document}

\renewcommand{\thefootnote}{\fnsymbol{footnote}}
\title{$C^1$-Classification of gapped parent Hamiltonians of quantum spin chains}

\author[S. Bachmann]{Sven Bachmann}
\address{Mathematisches Institut der Universit\"at M\"unchen, Theresienstrasse 39, D-80333
M\"unchen, Germany}
\email{sven.bachmann@math.lmu.de}

\author[Y. Ogata]{Yoshiko Ogata}
\address{Graduate School of Mathematical Sciences,
The University of Tokyo, Komaba, Tokyo, 153-8914, Japan}
\email{yoshiko@ms.u-tokyo.ac.jp}

\begin{abstract}
We consider the $C^1$-classification of gapped Hamiltonians introduced in~\cite{Fannes:1992vq, Nachtergaele:1996vc} as parents Hamiltonians of translation invariant finitely correlated states. Within this family, we show that the number of edge modes, which is equal at the left and right edge, is the complete invariant. The construction proves that translation invariance of the `bulk' ground state does not need to be broken to establish $C^1$-equivalence, namely that the spin chain does not need to be blocked.
\end{abstract}

\maketitle

\date{\today }

\newtheorem{thm}{Theorem}[section]
\newtheorem{lem}[thm]{Lemma}
\newtheorem{prop}[thm]{Proposition}
\newtheorem{cor}[thm]{Corollary}
\newtheorem{assum}[thm]{Assumption}
\newtheorem{rem}[thm]{Remark}
\newtheorem{defn}[thm]{Definition}
\newtheorem{clm}[thm]{Claim}

\newcommand{\lv}{\left \vert}
\newcommand{\rv}{\right \vert}
\newcommand{\lV}{\left \Vert}
\newcommand{\rV}{\right \Vert}
\newcommand{\la}{\left \langle}
\newcommand{\ra}{\right \rangle}
\newcommand{\ltm}{\left \{}
\newcommand{\rtm}{\right \}}
\newcommand{\lcm}{\left [}
\newcommand{\rcm}{\right ]}
\newcommand{\ket}[1]{\lv #1 \ra}
\newcommand{\bra}[1]{\la #1 \rv}
\newcommand{\lmk}{\left (}
\newcommand{\rmk}{\right )}
\newcommand{\al}{{\mathcal A}}
\newcommand{\md}{M_d({\mathbb C})}
\newcommand{\ali}[1]{{\mathfrak A}_{[ #1 ,\infty)}}
\newcommand{\alm}[1]{{\mathfrak A}_{(-\infty, #1 ]}}
\newcommand{\nn}[1]{\lV #1 \rV}
\newcommand{\br}{{\mathbb R}}
\newcommand{\dm}{{\rm dom}\mu}
\newcommand{\Ad}{\mathop{\mathrm{Ad}}\nolimits}
\newcommand{\Proj}{\mathop{\mathrm{Proj}}\nolimits}  
\newcommand{\RRe}{\mathop{\mathrm{Re}}\nolimits}
\newcommand{\RIm}{\mathop{\mathrm{Im}}\nolimits}
\newcommand{\Tr}{\mathop{\mathrm{Tr}}\nolimits}
\newcommand{\spn}{\mathop{\mathrm{span}}\nolimits}
\newcommand{\Mat}{\mathop{\mathrm{Mat}}\nolimits}
\newcommand{\GL}{\mathop{\mathrm{GL}}\nolimits}
\newcommand{\spa}{\mathop{\mathrm{span}}\nolimits}
\newcommand{\supp}{\mathop{\mathrm{supp}}\nolimits}
\newcommand{\rank}{\mathop{\mathrm{rank}}\nolimits}
\newcommand{\idd}{\mathop{\mathrm{id}}\nolimits}
\newcommand{\ran}{\mathop{\mathrm{Ran}}\nolimits}
\newcommand{\dr}{ \mathop{\mathrm{d}_{{\mathbb R}^k}}\nolimits} 
\newcommand{\dc}{ \mathop{\mathrm{d}_{\cc}}\nolimits} \newcommand{\drr}{ \mathop{\mathrm{d}_{\rr}}\nolimits} 
\newcommand{\zin}{\mathbb{Z}}
\newcommand{\rr}{\mathbb{R}}
\newcommand{\cc}{\mathbb{C}}
\newcommand{\nan}{\mathbb{N}}\newcommand{\bb}{\mathbb{B}}
\newcommand{\aaa}{\mathbb{A}}\newcommand{\ee}{\mathbb{E}}
\newcommand{\pp}{\mathbb{P}}
\newcommand{\wks}{\mathop{\mathrm{wk^*-}}\nolimits}
\newcommand{\he}{\hat {\mathbb E}}
\newcommand{\ikn}{{\caI}_{k,n}}
\newcommand{\mk}{\Mat_k(\cc)}
\newcommand{\hbb}{H^{k,\bb}_{m}}
\newcommand{\pbb}{\Phi^{k,\bb}_{m}}
\newcommand{\gbb}{\Gamma^{k,\bb}_{N,m}}
\newcommand{\mbb}{m^{k,\bb}}
\newcommand{\E}[1]{\widehat{\mathbb{E}}^{(#1)}}
\newcommand{\dist}{\dc}
\newcommand{\braket}[2]{\left\langle#1,#2\right\rangle}
\newcommand{\abs}[1]{\left\vert#1\right\vert}
\newtheorem{nota}{Notation}[section]
\def\qed{{\unskip\nobreak\hfil\penalty50
\hskip2em\hbox{}\nobreak\hfil$\square$
\parfillskip=0pt \finalhyphendemerits=0\par}\medskip}
\def\proof{\trivlist \item[\hskip \labelsep{\bf Proof.\ }]}
\def\endproof{\null\hfill\qed\endtrivlist\noindent}
\def\proofof[#1]{\trivlist \item[\hskip \labelsep{\bf Proof of #1.\ }]}
\def\endproofof{\null\hfill\qed\endtrivlist\noindent}

\renewcommand{\theenumi}{\roman{enumi}}


\section{Introduction}

A phase transition refers to a qualitative change in the properties of a family of physical systems as a parameter crosses a critical value, as for example the breaking of a continuous symmetry in thermal states as the temperature changes. The term \emph{quantum phase transition}~\cite{Sachdev:1999} refers somewhat unluckily to transitions happening at zero temperature, in particular qualitative changes in the ground states of quantum systems depending on a parameter. The archetypal example is here the transition from a unique ground state to a two-dimensional ground state space happening in the Ising model in a transverse magnetic field, which is accompanied by the closing of the spectral gap above the ground state energy. Such ground state phases and the transitions between them have received renewed attention recently both for fundamental reasons and for their potential applications in quantum information theory, with a particular focus on the structure of entanglement in the ground states.
It is a natural and important question to consider the classification of
gapped Hamiltonians, namely of Hamiltonians that have uniform spectral gap above the ground
state energy~\cite{Hastings:2005cs, Chen:2010gb, Chen:2011iq}.

In the context of quantum spin systems, a widely accepted criterion for the classification of gapped Hamiltonians
is as follows: two gapped Hamiltonians are equivalent if and only if they are connected by a continuous path of uniformly
gapped Hamiltonians
~\cite{Chen:2010gb, Chen:2011iq}.
In this article, we consider a bit stronger version of this, a {\it  $C^1$-equivalence}. We say
two gapped Hamiltonians are $C^1$-
equivalent
if and only if they are connected by a continuous {\it and piecewise $C^1$}-path of uniformly gapped Hamiltonians. We call the classification of 
gapped Hamiltonians with respect to this equivalence relation, 
the {\it $C^1$-classification of gapped Hamiltonians}.

In~\cite{Bachmann:2011kw}, 
it was shown that 
the `ground state structure' is an invariant of
this $C^1$-classification of gapped Hamiltonians.
The statement of ~\cite{Bachmann:2011kw}
is quite general and in particular does not refer to the spatial
dimension of the spin system. The need to prove the existence of a uniform spectral gap however makes the construction of relevant examples a hard problem, in particular in higher dimensions. In one dimension, the martingale method has been successfully applied to a large class of models, namely to systems with frustration free, finitely correlated ground states~\cite{Fannes:1992vq,Nachtergaele:1996vc}. They are simple, yet correlated states, and~\cite{Fannes:1992vq} gives a general recipe to construct gapped Hamiltonians which have a finitely correlated ground states, with a simple control of the spectral gap above the ground state energy. 

For translation-invariant one dimensional models, we consider the three possible infinite volume limits of finite volume ground states: on the bi-infinite chain -- the \emph{bulk ground states} -- and on the two possible half-infinite chains -- the \emph{left/right edge states}. By~\cite{Bachmann:2011kw}, the dimensions of these three ground state spaces are invariants of the $C^1$-classification. Of particular interest is the case of a unique, translation invariant bulk state, where the index is reduced to the pair of the numbers of edge modes. It is in general not clear if it is
the \emph{complete invariant}, namely if it uniquely determines the 
phase of the gapped Hamiltonians in the class. The Hamiltonians in ~\cite{Fannes:1992vq} have symmetric edge states, namely an equal number left and right edge modes, and the spin-$1$ antiferromagnetic model introduced by Affleck-Lieb-Kennedy-Tasaki in~\cite{Affleck:1988vr} is a physically relevant element of that class. On the other hand, the gapped models introduced in~\cite{Bachmann:2012uu, Bachmann:2012bf}, the `PVBS models', which also have a unique finitely correlated ground state in the bulk, have asymmetric edge modes and these Hamiltonians do not belong to the class given in ~\cite{Fannes:1992vq}.

The recurrent claim that all  such one-dimensional frustration free models
belong to the same phase and in particular that they are all equivalent to a pure product state, see e.g.~\cite{Schuch:2011ve}, was refined and partly clarified in~\cite{Bachmann:2012uu, Bachmann:2012bf} for PVBS models (but see~\cite{Wolf:2006cc} for explicit constructions of quantum phase transitions in families of finitely correlated states). There, it was shown that for the restricted class of PVBS models,
 the dimensions of edge states on the two possible half-infinite chains are indeed
complete invariant. It is further shown that the AKLT model belongs to one of these phases, namely the phase of a PVBS model with two-dimensional ground state space at each edge. 

In this paper, we prove that the number of edge modes is indeed the complete invariant of the $C^1$-classification within the family of gapped Hamiltonians of~\cite{Fannes:1992vq}. Explicitly, we construct a smooth path of uniformly gapped Hamiltonians between any two given elements of the family which have the same number of edge modes. Importantly, and unlike~\cite{Schuch:2011ve}, we do so \emph{without blocking sites}, and hence without breaking translation invariance of the ground state to a finite number of periodic states. Moreover, the interaction range along the path can be chosen to be constant and we give an explicit upper bound on the shortest such range. 
\\

\noindent{\bf Notations.}
We denote the Euclidean distance between a point $x$ and a subset $M$ in $\rr^k$ by 
$\dr(x,M)$. We also denote the Euclidean distance between two subsets $M_1$, $M_2$ in $\rr^k$ by 
$\dr(M_1,M_2)$.
Similarly,  we denote the Euclidean distance between a point $x$ and a subset $S$ in $\cc$ (resp. $\rr$) by 
$\dc(x,S)$ (resp. $\drr(x,S)$). 
For a subset $S$ of $\cc$ and $\delta>0$, the $\delta$-neighborhood of $S$ is denoted by $S_{\delta}$.
We denote the open ball in $\cc$ centered at $x\in\cc$ with radius $r$ by $B_{r}(x)$.
For a linear operator $T$ , we denote the spectrum of $T$ by $\sigma(T)$, and the spectral radius of $T$ by $r_T$.
For an isolated  subset $S$ of $\sigma(T)$, we denote the spectral projection 
of $T$ onto $S$ by $P_{S}^T$.
If $T$ is self-adjoint and $S$ is a subset of $\rr$, then $\Proj[T\in S]$ also indicates the spectral projection
of $T$ corresponding to $\sigma(T)\cap S$.
For $k\in\nan$, the set of orthogonal projections in
$k\times k$ matrices $\mk$ is denoted by ${\caP}(\mk)$
and the set of positive elements of $\mk$ by $\mk_+$.
We write $A>0$ for $A\in\mk$ if $A$ is strictly positive.
For $k\in\nan$,  $\Tr_{\mk}$ denotes the trace on $\mk$.
For a finite dimensional Hilbert space, braket $\braket{}{}$ denotes the
inner product of
the space under consideration.
For a Hilbert space $\mathfrak H$, we denote the set of all bounded liner operators on $\mathfrak H$ by
$B({\mathfrak H})$.

\section{The setup and main result}\label{sec:result}

For $\bbN\ni n\geq 2$, let $\caA$ be the finite dimensional C*-algebra $\caA = \Mat_n(\bbC)$, the algebra of $n\times n$ matrices. 
Throughout this article, this $n$ is fixed as the dimension of the spin under consideration.
We denote the set of all finite subsets in $\Gamma\subset{\bbZ}$ by ${\mathfrak S}_{\Gamma}$.
The number of elements in a finite set $\Lambda\subset {\bbZ}$ is denoted by
$|\Lambda|$.
When we talk about intervals in $\bbZ$, $[a,b]$ for $a\le b$,
means the interval in $\bbZ$, i.e., $[a,b]\cap \bbZ$.
We denote the set of all finite intervals in $\Gamma$
by ${\mathfrak I}_{\Gamma}$.
For each $z\in\bbZ$, we let $\caA_{\{z\}}$ be an isomorphic copy of $\caA$ and for any finite subset $\Lambda\subset\bbZ$, $\caA_{\Lambda} = \otimes_{z\in\Lambda}\caA_{\{z\}}$ is the local algebra of observables. 
For finite $\Lambda$, the algebra $\caA_{\Lambda} $ can be regarded as the set of all bounded operators acting on
a Hilbert space $\otimes_{z\in\Lambda}{\bbC}^n$.
We use this identification freely,
 and we denote the trace of $\caA_{\Lambda}\simeq B(\otimes_{z\in\Lambda}{\bbC}^n)$ by $\Tr_{\Lambda}$. 
Throughout this article, we fix an orthonormal basis $\{\psi_\mu\}_{\mu=1}^n$ of $\cc^n$.
If $\Lambda_1\subset\Lambda_2$, the algebra $\caA_{\Lambda_1}$ is naturally embedded in $\caA_{\Lambda_2}$ by tensoring its elements with the identity. Finally,
for an infinite subset $\Gamma$ of $\bbZ$, the algebra $\caA_{\Gamma}$ 
is given as the inductive limit of the algebras $\caA_{\Lambda}$ with $\Lambda\in{\mathfrak S}_{\Gamma}$. In particular,
$\caA_{\bbZ}$  is the chain algebra.
We denote the set of local observables in $\Gamma$ by $\caA_{\Gamma}^{\rm loc}=\bigcup_{\Lambda\in{\mathfrak S}_{\Gamma}}\caA_{\Lambda}
$.

For any $x\in\bbZ$, let $\tau_x$ be the shift operator by $x$ on $\caA_\bbZ$. 
An interaction is a map $\Phi$ from 
${\mathfrak S}_{\bbZ}$ into ${\caA}_{{\mathbb Z}}^{\rm loc}$ such
that $\Phi(X) \in {\caA}_{X}$ 
and $\Phi(X) = \Phi(X)^*$
for $X \in {\mathfrak S}_{\bbZ}$. 
An interaction $\Phi$ is translation invariant if
$
\Phi(X+j)=\tau_j\lmk
\Phi(X)\rmk,
$
for all $ j\in{\mathbb Z}$ and $X\in  {\mathfrak S}_{\bbZ}$.
Furthermore, it is of finite range if there exists an $m\in {\mathbb N}$ such that
$
\Phi(X)=0$,
for $X$ with diameter larger than $m$.
In this case, 
we say that the interaction length of $\Phi$ is less than or equal to $m$.
A Hamiltonian associated with $\Phi$ is a net of self-adjoint operators $H:=\left(H_\Lambda\right)_{\Lambda\in{\mathfrak I}_{\bbZ}}$ such that 
\begin{equation}\label{GenHamiltonian}
H_{\Lambda}:=\sum_{X\subset{\Lambda}}\Phi(X).
\end{equation}
Note that $H_{\Lambda}\in {\caA}_{\Lambda}$. Without loss of generality we consider positive interactions i.e.,  $\Phi(X)\geq 0$
for any $X\in{\mathfrak S}_{\bbZ}$, throughout this article.
We denote the set of all positive translation invariant finite range interactions by $\caJ$.
Futhermore, for $m\in\nan$, we denote by $\caJ_m$ the set of all positive translation invariant
interactions with interaction length less than or equal to $m$.

For a finite interval $\Lambda$, a ground state of
$H_{\Lambda}$ means a state on $\caA_{\Lambda}$
with support in the lowest eigenvalue space of $H_{\Lambda}$.
We denote the set of all ground states of $H_{\Lambda}$
on $\caA_{\Lambda}$ by $\caS_{\Lambda}(H)$.
For $\Lambda\in{\mathfrak I}_{\Gamma}$, any of the elements in 
 $\caS_{\Lambda}(H)$ can be extended to a 
state on $\caA_{\Gamma}$, and there exits a weak-$*$
accumulation points of such extensions,
in the thermodynamical limit $\Lambda\to\Gamma$.
We denote the set of all such accumulation points by
$\caS_{\Gamma}(H)$.

Let us specify what we mean by a \emph{gapped}  Hamiltonian:
\begin{defn}
A Hamiltonian $H:=\left(H_\Lambda\right)_{\Lambda\in{\mathfrak I}_{\bbZ}}$
associated with a positive translation invariant finite range interaction is \emph{gapped}  
if there exists $\gamma>0$ and $N_0\in\nan$ 
such that the difference 
between the smallest and the next-smallest eigenvalue of $H_\Lambda$, is bounded below by 
$\gamma$,  for all finite intervals $\Lambda\subset\bbZ$ with $|\Lambda|\ge N_0$. 
\end{defn}
We call this $\gamma$ a gap of the Hamiltonian $H$.
Note that the lowest eigenvalue may be degenerate.
\\

\noindent{\bf Gapped ground state phases.}
Now we introduce the $C^1$-classification 
of gapped Hamiltonians. 
We say  $\Phi:[0,1]\ni t\mapsto\Phi(t)\in  {\caJ}$ 
is a continuous and piecewise $C^1$-path if  for each $X\in{\mathfrak S}_{\bbZ}$,
$[0,1]\ni t\mapsto \Phi(t;X)\in {\caA}_X$ is continuous and piecewise $C^1$
with respect to the norm topology.
\begin{defn}[$C^1$-classification of gapped Hamiltonians]\label{def:phasec}
Let $H_0,H_1$ be gapped  Hamiltonians associated with interactions
$\Phi_{0},\Phi_{1}\in{\caJ}$.
We say that  $H_0,H_1$ are 
$C^1$-equivalent if the following conditions are satisfied.
\begin{enumerate}
\item There exists $m\in\nan$ and a continuous and piecewise $C^1$-path $\Phi:[0,1]\to {\caJ}_m$ such that $\Phi(0)=\Phi_0$, 
$\Phi(1)= \Phi_1$.
\item  Let   $H(t)$ be the Hamiltonian associated with $\Phi(t)$ for each $t\in[0,1]$.
There are $\gamma>0$, $N_0\in\nan$, and finite intervals $I(t)=[a(t), b(t)]$, whose endpoints $a(t), b(t)$ smoothly depending on $t\in[0,1]$,
such that for all finite intervals $\Lambda\subset\bbZ$ with $|\Lambda|\ge N_0$,
the smallest eigenvalue of $H(t)_\Lambda$ is in $I(t)$ and
the rest of the spectrum is in $[b(t)+\gamma,\infty)$.
\end{enumerate}
\end{defn}
The advantage of considering $C^1$-paths over just continuous ones is
as follows.
The following theorem is a special case of the Theorem 5.5 of ~\cite{Bachmann:2011kw}.
\begin{thm}[\cite{Bachmann:2011kw}]\label{qla}
Suppose that two gapped Hamiltonians $H_0,H_1$ are 
$C^1$-equivalent.
Then, for $\Gamma=(-\infty,-1]\cap \bbZ$, $\Gamma=[0,\infty)\cap\bbZ$ and $\Gamma=\bbZ$, there exists a quasi-local automorphism $\alpha_\Gamma$ of $\caA_\Gamma$ such that
\begin{equation*}
\caS_\Gamma(H_1) = \caS_\Gamma(H_0) \circ \alpha_\Gamma.
\end{equation*}
\end{thm}
See ~\cite{Bachmann:2011kw} for the general statement and the definition of quasi-locality. In other words, the structure of the bulk and of the left/right edge ground state spaces are invariants of the $C^1$-classification.
\\

\noindent{\bf Intersection property and parent Hamiltonians.}
Let $\caD_N$ be a subspace of $\otimes_{i=0}^{N-1}\cc^n$, for each $N\in \nan$.
We say that the sequence of subspaces $\{\caD_N\}_{N\in \nan}$ satisfies the {\it intersection property},
if there exists an $m\in\nan$, such that the relation 
\begin{equation}\label{intersection property}
\caD_N = \bigcap_{x=0}^{N-m} (\bbC^{n})^{\otimes x}\otimes \caD_{m}\otimes (\bbC^{n})^{\otimes N-m-x}, 
\end{equation} holds for all $N\ge m$.
In order to specify the number $m\in \nan$, we will say that $\{\caD_N\}_{N\in\nan}$ satisfies  Property~(I,$m$) 
when (\ref{intersection property}) holds for $m$
and all $N\ge m$.
Note that Property~(I,$m$) implies Property~(I,$m'$)  for all $m'\ge m$.

Given a sequence of nonzero spaces $\{\caD_N\}$ satisfying Property~(I,$m$),
there is a natural positive interaction for which $\caD_N$ are the ground state spaces
of the corresponding Hamiltonian.
Namely let $Q_m$ be the orthogonal projection onto the orthogonal complement of $\caD_m$ in $\otimes_{i=0}^{m-1}\cc^n$, and define
\begin{align*}
\Phi(X):=
\begin{cases}
\tau_x\lmk Q_m\rmk, & \text{if }  X=[x,x+m-1]  \text{ for some }  x\in\bbZ\\
0,&\text{otherwise}
\end{cases}\,.
\end{align*}
By (\ref{intersection property}), we see that $\ker H_{[0,N-1]}=\caD_N$,
for $N\ge m$, for the Hamiltonian $H=(H_{\Lambda})$ associated with $\Phi$.
We shall refer to that particular Hamiltonian as the \emph{parent Hamiltonian}
of $\{\caD_N\}$.
\\

\noindent{\bf Gapped Hamiltonians in~\cite{Fannes:1992vq}.}
Now we recall the class of Hamiltonians introduced
in ~\cite{Fannes:1992vq}. For $k\in\nan$, let 
$\caT_k$ be the set of all primitive completely positive maps on $\mk$
with spectral radius $1$.
It is well known that for each $T\in \caT_k$, $e_T:=P_{\{1\}}^T(1)$
is positive and invertible.
Furthermore, there exists a faithful state $\varphi_T$ given by
$P_{\{1\}}^T(a)=\varphi_T(a)e_T$ for $a\in\mk$.
There exists a positive invertible element $\rho_T$ in $\mk$ with
$\varphi_T=\Tr_{\mk}(\rho_T\cdot )$.
Note that $\varphi_T$ is $T$-invariant and $\varphi_T(e_T)=1$ (see Appendix \ref{sec:PathMaps}). For $T\in \caT_k$, we denote
$a_{T}:=\lV e_T^{-1}\rV$ and $c_{T}:=\lV \rho_T^{-1}\rV$.
Clearly, 
$0<a_{T},c_{T}<\infty$.

For each $n$-tuple of $k\times k$-matrices $\bb=(B_1,\cdots,B_n)\in\Mat_{1,n}(\mk)$, we define a completely positive map $\widehat \ee^{\bb}$
on $\mk$ by
\[
\widehat \ee^{\bb}:=\sum_{\mu=1}^n B_\mu\cdot B_\mu^*,
\]
and with this,
\begin{equation*}
B_{n,k} := \left\{
\bb=(B_1,\cdots,B_n)\in\Mat_{1,n}(\mk)\mid
\widehat \ee^{\bb}\in\caT_k
\right\}.
\end{equation*}

Now, for $N\in\nan$ and given $\bb\in \Mat_{1,n}(\mk)$, define 
$\Gamma_{N}^{k,\bb}
:\mk\to \otimes_{i=0}^{N-1}\cc^n$ by
\begin{equation}\label{Gammapq}
\Gamma_{N}^{k,\bb}(C):=\sum_{\mu_1,\ldots,\mu_N=1}^n
\Tr_{\mk} \lmk
CB_{\mu_N}^*\cdots B_{\mu_1}^* \rmk
\psi_{\mu_1}\otimes\cdots\otimes \psi_{\mu_N},\qquad
C\in\mk.
\end{equation}
Furthermore, set
$
{\caG}_{N}^{k,\bb}:=\ran \Gamma_{N}^{k,\bb},
$
and let $G^{k,\bb}_{N}$ be the orthogonal projection onto 
${\mathcal G}_{N}^{k,\bb}$ in $\otimes_{i=0}^{N-1}\cc^n$. 
This gives us a sequence of subspaces $\{{\mathcal G}_{N}^{k,\bb}\}_N$, and we shall say that $(k,\bbB)$ satisfies Property~(I,$m$) if the spaces $\{{\mathcal G}_{N}^{k,\bb}\}$ do so. For  $(k,\bbB)$, we set
\[
m^{k,\bb}:=\min \left\{m\in\nan\mid  (k,\bbB)\;{\rm satisfies} \; {\rm Property~(I},m)\right\}.
\]

Let ${\mathfrak E}_d$ denote the set of states on $\mathrm{Mat}_d(\bbC)$
for $d\in\nan$. 
\begin{prop}\label{thm:phases}
Let $k\in\bbN$, and $\bb\in B_{n,k}$.
\begin{enumerate}
\item $(k,\bbB)$ satisfies the intersection property.
\item For each $m\in\nan$, 
let $\Phi^{k,\bb}_{m}$ be a positive interaction  given by
\begin{align}
\Phi^{k,\bb}_{m}(X):=
\begin{cases}
\tau_x\lmk 1-G^{k,\bb}_{m}\rmk, & \text{if }  X=[x,x+m-1]  \text{ for some }  x\in\bbZ\\
0,&\text{otherwise}
\end{cases}\,,
\end{align}
and $H^{k,\bb}_{m}$ the Hamiltonian associated with 
$\Phi^{k,\bb}_{m}$.
Then for any $m\ge m^{k,\bb}$,
\begin{enumerate}
\item $H^{k,\bb}_{m}$ is gapped,
\item$\caS_{\bbZ}(H^{k,\bb}_{m})$ consists of  a unique state $\omega^{\bb}_{\infty}$ on $\caA_{\bbZ}$,
\item there exist affine bijections 
\[
\Xi_L:{\mathfrak E}_{k}\to \caS_{(-\infty,-1]}(H^{k,\bb}_{m}),\qquad
\Xi_R:{\mathfrak E}_{k}\to \caS_{[0,+\infty)}(H^{k,\bb}_{m}).
\]
 \end{enumerate}
 \end{enumerate}
\end{prop}
Let $\caH$ denote the set of all gapped Hamiltonians given by the recipe of above proposition. We note that $\caH$ corresponds to the Hamiltonians introduced in~\cite{Fannes:1992vq}.

\noindent{\bf Classification of the parent Hamiltonians.}
We now introduce a class of Hamiltonians labeled by $k\in\bbN$ as
\begin{equation*}
{\mathcal H}_{k}
:=\ltm
H^{k,\bb}_{m}\mid
\bb\in B_{n,k},\;
\bbN\ni m \geq m^{k,\bb}
\rtm.
\end{equation*}
By the definition, we have
\[
{\mathcal H}=\cup_{k\in\nan}{\mathcal H}_{k},
\]
and from Proposition~\ref{thm:phases}(iii),
${\mathcal H}_k$ is the class of Hamiltonians in $\mathcal H$
whose left and right edge state spaces are both of dimension $k$. 

The following theorem which is the main result of this paper is that this dimension
$k$ is the complete invariant for the $C^1$-classification within the set $\caH$:
\begin{thm}\label{mainclass}
For $k, k'\in\nan$, let
$H$, $H'$ be gapped Hamiltonians on $\caA_{\bbZ}$ such that
$H\in {\mathcal H}_{k}$ and 
$H'\in {\mathcal H}_{k'}$. 
Then  $H$ and $H'$ are 
$C^1$-equivalent gapped Hamiltonians
if and only if $k=k'$.
Furthermore, the path of interactions
in Definition~\ref{def:phasec}(i)
can be taken to be in $\caJ_m$, for
$m=\max\{M,M'\}$, where
$M,M'\in\nan$ are numbers which depend only on $H$, $H'$
respectively.
 \end{thm}
 
A classification of these Hamiltonians in ~\cite{Fannes:1992vq, Nachtergaele:1996vc}
is considered in ~\cite{Schuch:2011ve}, but there, the path is allowed to 
take {\it periodic} interactions, instead of {\it translation invariant} interactions.
Allowing periodic interactions corresponds to considering the situation $n\ge k^2$, which is much easier.
To find a path with uniform gap connecting 
$H^{k,\bb_0}_{m}$ and  $H^{k,\bb_1}_{m}$,
we need to find a path of $\widehat\bbE^{\bb(t)}$ satisfying 
$\bb(0)=\bb_0$, $\bb(1)=\bb_1$.
By the definition, these
completely positive maps $\widehat\bbE^{\bb(t)}$ should have Kraus rank less than or equal to $n$, and the gap remains
open if they are primitive.
But if $n\ge k^2$, the Kraus rank is less than or equal to $n$ for any kind of completely positive map on $\mk$.
Therefore, we may just take a path $(1-t)\widehat \bbE^{\bb_0}+t\widehat\bbE^{\bb_1}$, which is primitive because 
of the primitivity of $\widehat \bbE^{\bb_0}$ and $\widehat\bbE^{\bb_1}$, and with Kraus rank less than or equal to $k^2\le n$.
Translation invariance requires an additional work, because we can not take such a simple path: Section~\ref{sec:PathMatrices} is devoted to this problem, covering the case $k^2<n$.

This paper is organized as follows. In Section \ref{edge}, we summarize the ground state structure of Hamiltonians $H^{k,\bbB}_m$.
The classification Theorem~\ref{mainclass} is proven in Section~\ref{sec:PathHamiltonians}. 
The most nontrivial ingredient of the proof is the deformation of $\bb$, namely the pathwise connectedness of $B_{n,k}$, which we shall prove in Section~\ref{sec:PathMatrices}.

\section{The ground state structure of $H^{k,\bb}_{m}$}\label{edge}
In this section, we recall the spectral property of $H^{k,\bb}_{m}$ and
characterize the edge ground states on left and right half-infinite chains. For $k\in\nan$ and $\bb\in B_{n,k}$, we shall denote $e_\bb = e_{\widehat{\bbE}^\bb}$, $\varphi^\bb = \varphi_{\widehat{\bbE}^\bb}$, and 
$a^{\bb}=a_{\widehat\ee^{\bb}}$, 
$c^{\bb}=c_{\widehat\ee^{\bb}}$,
$r_{\bb}=r_{\widehat{\bbE}^\bb}$.

\subsection{The parent Hamiltonian and its spectral gap}\label{hamgap}
We first recall the results proven in~\cite{Fannes:1992vq}. Let $k\in\nan$ and $\bb\in B_{n,k}$. For each $N\in \nan$, let
 \begin{align*}
&E^{k,\bb}(N) :=k a^{\bb} c^\bb\left \Vert (\widehat\bbE^\bbB)^{N}(1-P^\bb_{\{1\}}) \right \Vert.
\\
&F^{k,\bb}:=
\frac{4}{a^{\bb} c^\bb}\lmk
\sup_N E^{k,\bb} (N)+c^\bb +a^{\bb}\Tr_{\mk} e_\bb
\rmk,\\
&L^{k,\bb}:=\min\left\{L\in\nan\;:\;
\sup_{N\ge L} E^{k,\bb}(N)<\frac 12\right\}.\\
&\bar l^{k,\bb}:=\min\ltm
l\in\nan \;:\; 
\sup_{N\ge l}\lmk 
\sqrt{N+1}\lmk
3E^{k,\bb}(N)F^{k,\bb}+2
\rmk
\lmk
E^{k,\bb}(N)F^{k,\bb}
\rmk
+ E^{k,\bb}(N)
\rmk<1\rtm.
\end{align*}
As $\widehat{\bbE}^\bb\in{\caT}_k$, its spectrum satisfies 
$\sigma(\widehat{\bbE}^\bb)\setminus\{1\}\subset \{z\in\cc\mid |z|<1\}$. Therefore, $E^{k,\bb}(N)$ goes to $0$ exponentially fast with respect to
$N$. It follows that $F^{k,\bb}$ $L^{k,\bb}$, $\bar l^{k,\bb}$ are well-defined and finite. By the definition, clearly, $L^{k,\bb}\le \bar l^{k,\bb}$.

\begin{prop}\label{Prop:Gapped pq}
Let $k\in\bbN$, and $\bb\in B_{n,k}$.
For each $m\in\nan$, 
let $\Phi^{k,\bb}_{m}$ be a positive interaction  given by
\begin{align}
\Phi^{k,\bb}_{m}(X):=
\begin{cases}
\tau_x\lmk 1-G^{k,\bb}_{m}\rmk,& \text{if }   X=[x,x+m-1]   \text{ for some }   x\in\bbZ\\
0,&\text{otherwise}
\end{cases}\,,
\end{align}
and $H^{k,\bb}_{m}$ be the translation invariant
Hamiltonian associated with 
$\Phi^{k,\bb}_{m}$.
Then for any $N,m\in\nan$ with
$N\ge m\ge m^{k,\bb}$,
the lowest eigenvalue of 
$
(H^{k,\bb}_{m})_{[0,N-1]}
$
is $0$ and the corresponding spectral projection is $G^{k,\bb}_{N}$.
Furthermore, for all $l\in\nan$ with $\max\{\bar l^{k,\bb},m\}< l<N$,
\begin{align*}
\frac{\gamma_{l,m}^{k,\bb}}{4(l+2)}\lmk 1-G^{k,\bb}_{N}\rmk
\le
(H^{k,\bb}_{m})_{[0,N-1]}.
\end{align*}
Here,
\[
\gamma_{l,m}^{k,\bb}=
\drr\lmk\sigma\lmk\lmk 
H^{k,\bb}_{m}\rmk_{[0,l-1]}
\rmk \setminus \{0\},0\rmk
\]
is the spectral gap of the finite volume Hamiltonian on an interval of length $l$.
\end{prop}
We refer to~\cite{Fannes:1992vq} for the proof of the proposition. One crucial ingredient is the proof of the Intersection Property for the sets $\caG_N^{k,\bbB}$, which itself is a consequence of the \emph{injectivity} of the map $\Gamma_N^{k,\bbB}$, see~(\ref{Gammapq}), for $N$ larger than some $N_0$. We give here an alternative proof of this fact, which yields a quantitative upper bound on $N_0$, and that will be useful later.

\subsection{The injectivity and the intersection property}
For any $n,k,m\in\nan$ and $\bb\in \Mat_{1,n}(\Mat_{k}(\cc))$, let ${\mathcal K}_m(\bb)$ be the following span of monomials of degree $m$ in the $B_\mu$'s,
\begin{equation}\label{KB}
{\mathcal K}_m(\bb) :=\spn\left\{B_{\mu_m}B_{\mu_{m-1}}\ldots B_{\mu_1}\mid
(\mu_1,\ldots,\mu_m)\subset\{1,\ldots n\}^{\times m}\right\}.
\end{equation}

Furthermore, let
\begin{equation}\label{Xnkm}
X_{n,k,m}:=
\left\{
\bb\in \Mat_{1,n}(\Mat_{k}(\cc))\mid {\mathcal K}_m(\bb)=\Mat_{k}(\cc)
\right\}.
\end{equation}
\begin{lem}\label{lem:intersection}
Let $k\in\nan$, and $\bb\in \Mat_{1,n}(\mk)$.
Then 
the followings are equivalent.
\begin{enumerate}
\item[(i)]
$r_{\bb}>0$ and $r_{\bb}^{-\frac 12}\bb \in B_{n,k}$,
\item[(ii)]
there exists an $m\in\nan$ such that $\bb\in X_{n,k,m'}$, for all
$m'\ge m$,
\item[(iii)]
there exists an $m\in\nan$ such that $\bb\in X_{n,k,m}$.
\end{enumerate}
Furthermore, if these (equivalent) conditions hold,
set
\[
s^{k,\bb}:=
\min\{
m\in \nan\mid \bb\in X_{n,k,m}
\}.
\]
Then,
\begin{enumerate}
\item[(a)]
for any $s\ge s^{k,\bb}$, ${\mathcal K}_s(\bb)=\mk$,
\item[(b)]
for any $s\ge s^{k,\bb}$,
$\Gamma_{s}^{k,\bb}$ is injective, and
$(k,\bbB)$ satisfies Property~(I,$s+1$),
\item[(c)]
$s^{k,\bb}\leq  k^4$
\item[(d)]if $B_1$ is invertible, then $s^{k,\bb}\leq k^2$.
\end{enumerate}
\end{lem}
\proof
To prove the first part, we note that any of (i), (ii), (iii) implies the irreducibility of $\widehat\bbE^\bb$.
For if (i) holds, then for any orthogonal projection $P\in{\mathcal P}(\mk)$ with
$\widehat\bbE^\bb(P\mk P)\subset P\mk P$, 
we have
\[
\varphi_{(r_{\bb})^{-1}\widehat\bbE^\bb}(P)e_{(r_{\bb})^{-1}\widehat\bbE^\bb}=\lim_{N\to\infty} r^{-N}_\bb(\widehat\bbE^\bb)^N(P)\in P\mk P.
\]
As $(r_{\bb})^{-1}\widehat\bbE^\bb\in\caT_k$, if $P\neq 0$, the left hand side is strictly positive, hence $P=1$.
This means the irreducibility of $\widehat\bbE^\bb$.
If (iii) holds, it is easy to check that $(\widehat\bbE^\bb)^m(A)>0$ for any nonzero $A\in\mk_+$.
Hence, for any $t>0$ and nonzero $A\in\mk_+$, we have $\exp(t\widehat\bbE^\bb)(A)>0$, hence $\widehat\bbE^\bb$ is irreducible
(see Theorem \ref{pm1}).
As (ii) implies (iii), this case is clear. 

Therefore, throughout the proof, we may assume that $\widehat\bbE^\bb$ is irreducible.
In this case, the spectral radius $r_{\bb}$ is strictly positive and a nondegenerate eigenvalue of
$\widehat\bbE^\bb$ with some strictly positive eigenvector $h_\bb$ (see Theorem~\ref{pm2}).
Hence,
\[
\widehat T_\bb:=(r_{\bb})^{-1}h_\bb^{-\frac 12} \widehat\bbE^\bb\lmk 
{h_\bb}^{\frac 12}(\cdot){h_\bb}^{\frac 12}
\rmk{h_\bb}^{-\frac 12}
\]
is a well-defined unital completely positive map.
Furthermore, we set \[
\bar \bb :=
\lmk (r_{\bb})^{-\frac 12}h_\bb^{-\frac 12} B_1 h_\bb^{\frac 12}, (r_{\bb})^{-\frac 12}h_\bb^{-\frac 12} B_2 h_\bb^{\frac 12},\ldots,
(r_{\bb})^{-\frac 12}h_\bb^{-\frac 12} B_n h_\bb^{\frac 12}\rmk,
\]
and note that $\widehat T_\bb=\widehat\bbE^{\bar \bb}$.

By Theorem \ref{pm3}, the followings are equivalent:
\begin{enumerate}
\item[(i)'] There exists a unique faithful $\widehat T_\bb$-invariant state $\psi_{\bb}$, and it satisfies
 \[
\lim_{l\to\infty}(\widehat T_\bb)^l(A)=\psi_\bb(A) 1,\qquad
A\in \Mat_k(\cc),
\]
\item[(ii)']
there exists an $m\in\nan$ such that ${\mathcal K}_{m'}(\bar \bb)=\mk$, for all
$m'\ge m$,
\item[(iii)']
there exists an $m\in\nan$ such that ${\mathcal K}_{m}(\bar \bb)=\mk$.
\end{enumerate}
Clearly, (ii), resp. (iii), is equivalent to (ii)', resp. (iii)'.
Hence it suffices to show that  (i) is equivalent to (i)'.

As (i) implies $(r_{\bb})^{-1}\widehat\bbE^\bb\in {\mathcal T}_k$,
there exists a $(r_{\bb})^{-1}\widehat\bbE^\bb$-invariant faithful state $\varphi^{(r_{\bb})^{-\frac 12}\bb}$ on $\mk$,
and we have
\[
\lim_{N\to\infty} ( \widehat T_\bb)^N(A)
=\psi_\bb(A) 1,
\qquad A\in\mk.
\]
Here, 
\[
\psi_\bb(A) 
=\frac{\varphi^{(r_{\bb})^{-\frac 12}\bb}\lmk
h_\bb^{\frac 12} A h_\bb^{\frac 12}
\rmk}{\varphi^{(r_{\bb})^{-\frac 12}\bb}(h_\bb)},\qquad A\in\mk
\]
is the unique $\widehat T_\bb$-invariant faithful state.

On the other hand, if (i)' holds, then as $\widehat T_\bb$ is similar to $(r_{\bb})^{-1}\widehat\bbE^\bb$,
 the spectral radius of $(r_{\bb})^{-1}\widehat\bbE^\bb$ is $1$,
and it is a strictly positive  non degenerate eigenvalue of  $(r_{\bb})^{-1}\widehat\bbE^\bb$.
Furthermore, we have
$\sigma((r_{\bb})^{-1}\widehat\bbE^\bb)\setminus \{1\}=\sigma(\widehat T_\bb)\setminus\{1\}\subset
\{z\in{\mathbb C}\;:\; |z|<1\}$, where the last inclusion is by the primitivity of the unital completely positive map $\widehat T_\bb$
(see Theorem \ref{pm3}).
Hence we have $(r_{\bb})^{-1}\widehat\bbE^\bb\in{\caT}_k$.
By the similarity of $\widehat T_\bb$ with $(r_{\bb})^{-1}\widehat\bbE^\bb$,
we have
\[
P_{\{1\}}^{(r_{\bb})^{-1}\widehat\bbE^\bb}(A)=\psi_{\bb}\lmk h_\bb^{-\frac 12} A h_\bb^{-\frac 12}\rmk h_\bb,\qquad
A\in \mk.
\]
From this, $e_{{(r_{\bb})^{-1}\widehat\bbE^\bb}}=P_{\{1\}}^{(r_{\bb})^{-1}\widehat\bbE^\bb}(1)
=\psi_{\bb}\lmk h_\bb^{-1} \rmk h_\bb$, is invertible in $\mk$. Furthermore, we have 
\[
P_{\{1\}}^{(r_{\bb})^{-1}\widehat\bbE^\bb}(A)=
\frac{\psi_\bb(h_\bb^{-\frac 12} A h_\bb^{-\frac 12})}{\psi_\bb(h_\bb^{-1})}
e_{{(r_{\bb})^{-1}\widehat\bbE^\bb}},\qquad
A\in \mk.
\]
Here, the state 
\[
\varphi^{(r_{\bb})^{-\frac 12}\bb}=\frac{\psi_\bb(h_\bb^{-\frac 12} \cdot h_\bb^{-\frac 12})}{\psi_\bb(h_\bb^{-1})}
\] is faithful on $\mk$, because $\psi_\bb$ is faithful.

Next we consider the latter half.
To see (a), assume that $\caK_s(\bb)\neq\mk$ for some $s>s^{k,\bb}$.
This means that there exists a nonzero $A\in\mk$ such that $\Tr_{\mk}(A^*B_{\mu_s}\cdots B_{\mu_{1}})=0$
for any
$\mu_1,\ldots,\mu_s\in\{1,\cdots,n\}$.
As $\caK_{s^{k,\bb}}(\bb)=\mk$, this implies $A^*B_{\mu_s}\cdots B_{\mu_{s^{k,\bb}+1}}=0$, for any 
$\mu_{s^{k,\bb}+1},\ldots,\mu_s\in\{1,\cdots,n\}$.
Hence we have $0=A^*(\widehat\bbE)^{s-s^{k,\bb}}(e_{\bb})=(r_{\bb})^{s-s^{k,\bb}}A^* e_{\bb}$.
Because $e_{\bb}$ is invertible, this means $A=0$, a contradiction.

To see the injectivity stated  in (b) for $s\ge s^{k,\bb}$, suppose that 
$\Gamma_{s}^{k,\bb}(C)=0$, for $C\in \mk$. Then, 
$\Tr_{\mk}(C B^*_{\mu^{(s)}})=0$, for any $s$-tuple $\mu^{(s)}\in\{1,\dots,n\}^{\times s}$.
As $s\ge s^{k,\bb}$, ${\mathcal K}_s(\bb)=\mk$ by (a), hence this implies $C=0$, and the injectivity holds.
Property~(I,$s+1$) in (b) can then be checked as in ~\cite{Fannes:1992vq}. Finally, (c), (d) is the quantum Wielandt's inequality in the least optimal case of only one linearly independent
Kraus operator~\cite{Sanz:2010aa}.
\endproof

\subsection{Edge ground states}\label{sec:edge}
Now we consider in detail the edge ground states, namely the half-chain analog of~\cite{Fannes:1992vq}.
Throughout this section we fix an orthonormal basis $\{e_{\alpha}\}_{\alpha=1}^k$ of $\cc^k$.
Furthermore, we define a sesquilinear form 
$\braket{}{}_{\bb}$ on $ \mk$ by
\begin{equation*}
\braket{B}{C}_{\bb} := \varphi^\bb (B\str e_\bb C),\quad B,C\in  \mk.
\end{equation*}
As $\varphi^\bb$ is faithful on $\mk$ and $e_\bb$ is invertible in $\mk$, this gives an inner product on
$\mk$. Furthermore,
\begin{equation}\label{bieq}
\Tr_{\mk}(X\str X)\leq a^{\bb} c^\bb \braket{X}{X}_{\bb},\quad
\Tr_{\mk}(X\str e_\bb X)\leq  c^\bb \braket{X}{X}_{\bb},\quad
\varphi^\bb(X\str X)\leq a^{\bb} \braket{X}{X}_{\bb},
\end{equation}
for any $X\in\mk$.

The following estimate can be found in~\cite{Fannes:1992vq}. We repeat its proof for completeness here.
\begin{lem}\label{lemma:injectivity}
Let $\bb\in B_{n,k}$.
 Then for any $B,C\in \mk$ and $N\in \nan$,
\begin{equation*}
\left\vert \braket{\Gamma_{N}^{k,\bb}(B)}{\Gamma_{N}^{k,\bb}(C)} - \braket{B}{C}_{\bb} \right\vert \leq E^{k,\bb}(N)\braket{B}{B}^{1/2}_{\bb}\braket{C}{C}_{\bb}^{1/2}.
\end{equation*}
Here $\braket{}{}$indicates the inner product of $\bigotimes_{i=0}^{N-1}\bbC^n$.
In particular, for any $N\in\nan$,
\begin{equation}\label{nb}
(1-E^{k,\bb}(N))\braket{B}{B}_{\bb}
\le \lV \Gamma_{N}^{k,\bb}(B)\rV^2
\le (1+E^{k,\bb}(N))\braket{B}{B}_{\bb},\quad B\in \mk,
\end{equation}and
$\Gamma_{N}^{k,\bb}$ is injective if
$E^{k,\bb}(N)<1$.
\end{lem}
\begin{proof}
First, recall that $P^\bb_{\{1\}}(a)=\varphi^\bb(a)e_{\bb}$. Then
\begin{equation*}
\braket{\Gamma_{N,p,q}^{k,\bb}(B)}{\Gamma_{N,p,q}^{k,\bb}(C)} =
 \sum_{\alpha,\beta=1}^k\left[\braket{{e_{\alpha}}}{(\widehat\bbE^\bbB)^{N}(1-P^\bb_{\{1\}})\left(B\str \ket{{e_{\alpha}}}\bra{{e_{\beta}}} C \right){e_{\beta}}} + \varphi^\bb(B\str \ket{{e_{\alpha}}}\bra{{e_{\beta}}} C) \braket{{e_{\alpha}}}{e_\bb{e_{\beta}}}
\right].
\end{equation*}
The second term is equal to $\varphi^\bb (B\str e_\bb C)=\braket{B}{C}_{\bb}$. The first term can be bounded from above by
\begin{equation*}
\sum_{\alpha,\beta=1}^k \Vert {e_{\alpha}} \Vert \Vert (\widehat\bbE^\bbB)^{(N)}(1-P^\bb_{\{1\}})  \Vert \Vert B\str {e_{\alpha}} \Vert \Vert C\str {e_{\beta}} \Vert \Vert {e_{\beta}} \Vert \le
k \Vert (\widehat\bbE^\bbB)^{(N)} (1-P^\bb_{\{1\}})  \Vert \Tr_{\mk} (B\str B)^{1/2}\Tr_{\mk} (C\str C)^{1/2}
\end{equation*}
using the Cauchy-Schwarz inequality for $\cc^k$. The first part of the lemma follows from this combined with the observation (\ref{bieq}). The second inequality (\ref{nb}) can be immediately checked from the first inequality. Finally, if $E^{k,\bb}_{p,q}(N)<1$, then from (\ref{nb}) , $\Gamma_{N,p,q}^{k,\bb}(B)=0$ implies
$\braket{B}{B}_{\bb}=0$. As $\braket{}{}_{\bb}$
is an inner product, this means $B=0$.
Therefore, $\Gamma_{N,p,q}^{k,\bb}$ is injective.
\end{proof}

For integers $b\le a$ and
an $a-b+1$-tuple $\mu^{(a-b+1)}=(\mu^{(a-b+1)}_1,\ldots,\mu^{(a-b+1)}_{a-b+1})\in\{1,\ldots,n\}^{\times (a-b+1)}$,
$\bar \psi_{ \mu^{(a-b+1)}}^{[b,a]}$ indicates the vector $\psi_{\mu_1^{(a-b+1)}}\otimes \psi_{\mu_{2}^{(a-b+1)}}\otimes\cdots 
\otimes \psi_{\mu_{a-b+1}^{(a-b+1)}}$
in $\otimes_{i=b}^{a}\cc^n$.

\begin{lem}\label{rba}
Let $\bb\in B_{n,k}$ and let $R^{\bb}:{\mathcal A}^{\rm loc}_{[0,\infty)}\to \mk$ be defined by
\begin{equation*}
R^{\bb}(A):=\sum_{\mu^{(a)},\nu^{(a)}\in \{1,\ldots,n\}^{\times a}}
\braket{\bar \psi^{[0,a-1]}_{\mu^{(a)}}}{A\bar\psi^{[0,a-1]}_{\nu^{(a)}}}
B_{\mu^{(a)}} e_\bb B_{\nu^{(a)}}\str,
\end{equation*}
if $A\in {\mathcal A}_{[0,a-1]}$ for $a\in\nan$.
Then $R^{\bb}$ is well-defined and extends to a completely positive map from the half-infinite chain ${\mathcal A}_{[0,\infty)}$
onto $\mk$, which we will denote by the same symbol $R^\bb$.
Furthermore, for any $A\in{\mathcal A}^{\rm loc}_{[0,\infty)}$ and $C\in \mk$, we have
\[
\lim_{N\to\infty}\braket{\Gamma_{N}^{k,\bb}(C)}{A\Gamma_{N}^{k,\bb}(C)}
=\varphi^\bb \lmk C^*R^{\bb}(A) C\rmk.
\]\end{lem}
\proof
Note that from the relation $\sum_{\mu=1}^n B_{\mu}e_\bb B_{\mu}^*=e_\bb$,
$R^\bb$ is well-defined.
It can be checked directly that $R^\bb\vert_{{\caA}_{[0,a-1]}}$ defines
a completely positive map on ${\caA}_{[0,a-1]}$ 
with norm $\lV R^\bb\vert_{{\caA}_{[0,a-1]}}\rV=\lV R^\bb(1)\rV=\lV e_\bb\rV$, for any $a\in\nan$.
Therefore, we can extend it to 
a completely positive map on ${\mathcal A}_{[0,\infty)}$.

To see that $R^\bb$ is surjective, recall from Lemma \ref{lem:intersection}, that for $\bb\in B_{n,k}$ 
and $s^{k,\bb}\le a\in\nan$, we have $\caK_a(\bb)=\mk$.
Therefore, for any $\xi_1,\xi_2\in\cc^k$ and $\eta\in p\cc^k\setminus \{0\}$,
there exist $\{\alpha_{\mu^{(a)}}\}_{\mu^{(a)}\in\{1,\ldots,n\}^{\times a}}\subset \cc$
$\{\beta_{\nu^{(a)}}\}_{\nu^{(a)}\in\{1,\ldots,n\}^{\times a}}\subset \cc$
\[
\sum_{\mu^{(a)}}\alpha_{\mu^{(a)}} B_{\mu^{(a)}}
=\frac{1}{\braket{\eta}{e_\bb\eta}}\ket{\xi_1}\bra{\eta},\qquad
\sum_{\nu^{(a)}}\beta_{\nu^{(a)}} B_{\nu^{(a)}}^*
=\ket{\eta}\bra{\xi_2}.
\]
Then for
\[
A=\sum_{\mu^{(a)},\nu^{(a)}}\alpha_{\mu^{(a)}} \beta_{\nu^{(a)}} 
\ket{\bar \psi^{[0,a-1]}_{\mu^{(a)}}}\bra{\bar \psi^{[0,a-1]}_{\nu^{(a)}}} ,
\]
we have $R^\bb(A)=\ket{\xi_1}\bra{\xi_2}$. Hence, $R^\bb$ is surjective.

For the latter equality, if $A\in \caA_{[0,a-1]}$,
as in Lemma \ref{lemma:injectivity} , we have
\begin{align*}
&\braket{\Gamma_{N}^{k,\bb}(C)}{A\Gamma_{N}^{k,\bb}(C)}\\
&=
 \sum_{\substack{\mu^{(a)},\nu^{(a)}\in \{1,\ldots,n\}^{\times a} \\ \mu^{(N-a)}\in \{1,\ldots,n\}^{\times (N-a)}}}
\braket{\bar\psi_{\mu^{(a)}}^{[0,a-1]}}{A {\bar\psi_{\nu^{(a)}}^{[0,a-1]}}}
 \Tr_{\mk}\lmk C^* B_{\mu^{(a)}}B_{\mu^{(N-a)}}\rmk
\Tr_{\mk} \lmk B_{\mu^{(N-a)}}^* B_{\nu^{(a)}}^* C\rmk
\\
&=
\sum_{\alpha,\beta=1}^k \sum_{\mu^{(a)},\nu^{(a)}\in \{1,\ldots,n\}^{\times a}}
\braket{e_{\alpha}}
{\lmk \widehat \bbE^{\bb}\rmk^{N-a} \lmk C^*  B_{\mu^{(a)}} \ket{e_\alpha}\bra{e_{\beta}} B_{\nu^{(a)}}^*  C\rmk e_{\beta}}\braket{\bar\psi_{\mu^{(a)}}^{[0,a-1]}}{A {\bar\psi_{\nu^{(a)}}^{[0,a-1]}}}\\
&\to 
\sum_{\alpha,\beta=1}^k \sum_{\mu^{(a)},\nu^{(a)}\in \{1,\ldots,n\}^{\times a}}
\braket{e_{\alpha}}{e_\bb 
 e_{\beta}}\varphi^\bb \lmk C^*  B_{\mu^{(a)}} \ket{e_\alpha}\bra{e_{\beta}} B_{\nu^{(a)}}^*  C\rmk\braket{\bar\psi_{\mu^{(a)}}^{[0,a-1]}}{A {\bar\psi_{\nu^{(a)}}^{[0,a-1]}}}\\
 &=\sum_{\mu^{(a)},\nu^{(a)}\in \{1,\ldots,n\}^{\times a}}\varphi^\bb \lmk C^*  B_{\mu^{(a)}} e_\bb B_{\nu^{(a)}}^*  C\rmk\braket{\bar\psi_{\mu^{(a)}}^{[0,a-1]}}{A {\bar\psi_{\nu^{(a)}}^{[0,a-1]}}}
 =\varphi^\bb\lmk C^* R^\bb (A) C\rmk .
\end{align*}
\endproof
Similarly,
\begin{lem}
Let $\bb\in B_{n,k}$ and let $L^{\bb}:{\mathcal A}^{\rm loc}_{(-\infty,-1]}\to \mk$ be defined by
\begin{equation*}
L^{\bb}(A):=\sum_{\mu^{(b)},\nu^{(b)}\in \{1,\ldots,n\}^{\times {b}}}
\braket{\bar \psi^{[-b,-1]}_{\mu^{(b)}}}{A\bar \psi^{[-b,-1]}_{\nu^{(b)}}}
B_{\nu^{(b)}}\str \rho^\bb B_{\mu^{(b)}},
\end{equation*}
if $A\in {\mathcal A}_{[-b,-1]}$ for $b\in\nan$. Then $L^{\bb}$ is well-defined and 
extends to a completely positive map from the half-infinite chain ${\mathcal A}_{(-\infty,-1]}$
onto $\mk$, which we will denote by the same symbol $L^\bb$.
Furthermore, for any $A\in{\mathcal A}^{\rm loc}_{(-\infty,-1]}$ and $C\in \mk$, we have
\[
\lim_{N\to\infty}\braket{\Gamma_{N}^{k,\bb}(C)}{\tau_N(A)\Gamma_{N}^{k,\bb}(C)}
=\Tr_{\mk}\lmk  e_\bb CL^\bb (A) C^* \rmk.
\]
\end{lem}

If $\omega$ is a state on $\mk$, 
\begin{equation*}
\omega^\bb_R(A) := \omega\lmk e_\bbB^{-1/2}R^\bb(A) e_\bbB^{-1/2}\rmk,\qquad A\in{\mathcal A}_{[0,\infty)},
\end{equation*}
defines a state on ${\mathcal A}_{[0,\infty)}$. Similarly, if $\omega$ is a state on $\mk$, then
\begin{equation*}
\omega^\bb_L(A) := {\omega\lmk(\rho^\bbB)^{-1/2}L^\bb(A)(\rho^\bbB)^{-1/2}\rmk}
,\qquad A\in{\mathcal A}_{(-\infty,-1]},
\end{equation*}
defines a state on ${\mathcal A}_{(-\infty,-1]}$.\\

We denote the sets of these states by
\begin{align}
\caE_{R}^{k,\bb} &:=\{\omega^\bb_R \;:\; \omega\text{ is a state on }\mk\}, \label{ER}\\
\caE^{k,\bb}_{L} &:=\{\omega^\bb_L \;:\; \omega\text{ is a state on }\mk\}. \label{EL}
\end{align}

Recall that ${\mathfrak E}_d$ denotes the set of states on $\mathrm{Mat}_d(\bbC)$
for $d\in\nan$. 
\begin{lem}\label{aim}
Let $k\in \nan$ and $\bb\in B_{n,k}$.
Then 
the maps $\Xi_R : {\mathfrak E}_{k}\to \caE_{R}^{k,\bb}$,
$\Xi_L : {\mathfrak E}_{k}\to \caE_{L}^{k,\bb}$ defined by
\begin{align*}
\Xi_R(\omega):=\omega^\bb_R,\qquad \Xi_L(\omega):=\omega^\bb_L,
\end{align*}
are affine bijections.
\end{lem}
\proof
From the definition, it is clear that $\Xi_L$, $\Xi_R$ are affine surjection.
As $R^\bb$ (resp. $L^\bb$) is surjective and $e_{\bb}$ (resp. $\rho^{\bb}$)
is bijective in $\mk$,
$\Xi_R$ (resp. $\Xi_L$) is injective.
\endproof
As in \cite{Fannes:1992vq}, we can show the following:
\begin{lem}
Let $\bb\in B_{n,k}$ and let $\omega_{\infty}^\bb:{\mathcal A}^{\rm loc}_{\bbZ}\to \bbC$ be defined by
\begin{equation*}\label{bbif}
\omega_{\infty}^{\bb}(A):=\sum_{\mu^{(a-b+1)},\nu^{(a-b+1)}\in \{1,\ldots,n\}^{\times a-b+1}}
\braket{\bar \psi^{[b,a]}_{\mu^{(a-b+1)}}}{A\bar \psi^{[b,a]}_{\nu^{(a-b+1)}}}
 \varphi^\bb \lmk B_{\mu^{(a-b+1)}} e_\bb B_{\nu^{(a-b+1)}}\str\rmk
\end{equation*}
if $A\in {\mathcal A}_{[b,a]}$ with $b\le a$. Then $\omega_{\infty}^\bb$  is well-defined and 
extends to a
state on the chain ${\mathcal A}_{\bbZ}$. Furthermore, 
for any $A\in{\mathcal A}^{\rm loc}_{\bbZ}$, and $C\in\mk$, we have
\[
\lim_{N,M\to\infty}\braket{\Gamma_{N+M}^{k,\bb}(C)}{\tau_{N}(A)\Gamma_{N+M}^{k,\bb}(C)}
=\omega_{\infty}^\bb(A)\varphi^\bb \lmk C^* e_\bb C\rmk.
\]
\end{lem}

On the other hand, for $\Phi^{k,\bb}_{m} = 1-G^{k,\bb}_{m}$, $m\in\nan$, recall that
$\caS_{[0,\infty)}(H^{k,\bb}_{m})$,
$\caS_{(-\infty,-1]}(H^{k,\bb}_{m})$, $\caS_{\bbZ}(H^{k,\bb}_{m})$
are the set of all $wk*$-accumulation points of ground states in finite intervals.
\begin{lem}
Let $k\in \nan$
and $\bb\in B_{n,k}$.
Then for $m\ge m^{k,\bb}$,
\begin{align*}
\caS_{[0,\infty)}(H^{k,\bb}_{m})&=
\{
\omega\;:
\;
\omega \text{ is a state on }{\mathcal A}_{[0,\infty)}\text{ such that }\omega\circ\tau_a(\Phi^{k,\bb}_{m})=0,\text{ for all }0\le a\in\zin
\}\\
\caS_{(-\infty,-1]}(H^{k,\bb}_{m})&=
\{
\omega\;:
\;
\omega \text{ is a state on }{\mathcal A}_{(-\infty,-1]}\text{ such that }\omega\circ\tau_{-b}(\Phi^{k,\bb}_{m})=0,
\text{ for all } m\le b\in\zin
\},\\
\caS_{\bbZ}(H^{k,\bb}_{m})&=
\{
\omega\;:
\;
\omega \text{ is a state on }{\mathcal A}_{\bbZ}\text{ such that }\omega\circ\tau_a(\Phi^{k,\bb}_{m})=0,\text{ for all }a\in\zin
\},
\end{align*}
\end{lem}
\proof
If $\omega$ is a state on ${\mathcal A}_{[0,\infty)}$ such that $\omega\circ\tau_{a}(\Phi^{k,\bb}_{m})=0$,
for all $0\le a\in\zin$.
Then its restriction $\omega\vert_{\caA_{\Lambda}}$ to each interval $\Lambda\subset [0,\infty)$
is a ground state of $(H^{k,\bb}_{m})_{\Lambda}$.
Hence $\omega$ is a $wk*$-accumulation point of extensions of 
$\omega\vert_{\caA_{\Lambda}}\in {\caS}_{\Lambda}((H^{k,\bb}_{m})_{\Lambda})$, hence
$\omega\in {\caS}_{[0,\infty)}(H^{k,\bb}_{m})$, by definition.
On the other hand, if $\omega\in {\caS}_{[0,\infty)}(H^{k,\bb}_{m})$, then
there exits a subnet $\{\Lambda'\}$ of intervals in $[0,\infty)$ associated with states $\omega_{\Lambda'}$
on ${\mathcal A}_{[0,\infty)}$ 
such that $\omega_{\Lambda'}\vert_{\caA_{\Lambda'}}\in \caS_{\Lambda'}((H^{k,\bb}_{m})_{\Lambda'})$,
and $\omega={\rm wk}*-\lim_{\Lambda'} \omega_{\Lambda'}$.
Hence we have $\omega\circ\tau_a(\Phi^{k,\bb}_{m})=\lim_{\Lambda'} \omega_{\Lambda'}(\tau_a(\Phi^{k,\bb}_{m}))=0$,
for all $0\le a\in\bbZ$.
\endproof
\begin{prop}\label{Prop:edge spaces}
Let $k\in \nan$ 
and $\bb\in B_{n,k}$. Then for $m\ge m^{k,\bb}$,
\begin{equation*}
\caS_{[0,\infty)}(H^{k,\bb}_{m})= \caE_{R}^{k,\bb},\qquad \caS_{(-\infty,-1]}(H^{k,\bb}_{m})= \caE^{k,\bb}_{L},\qquad
\caS_{\bbZ}(H^{k,\bb}_{m})=\{\omega_{\infty}^\bb\},
\end{equation*}
where $\caE_{R}^{k,\bb},\caE_{L}^{k,\bb}$ were defined in~(\ref{ER},\ref{EL}).
 In particular, 
the maps
\begin{equation*}
\Xi_R : {\mathfrak E}_{k}\to \caE_{R}^{k,\bb}=\caS_{[0,\infty)}(H^{k,\bb}_{m}),\qquad
\Xi_L : {\mathfrak E}_{k}\to \caE_{L}^{k,\bb}=\caS_{(-\infty,-1]}(H^{k,\bb}_{m})
\end{equation*}
are affine bijections.
\end{prop}
\begin{proof}
First we show $\caE_{R}^{k,\bb}\subset \caS_{[0,\infty)}(H^{k,\bb}_{m})$.
Let $\omega$ be a state on $\mk$ given by a density matrix
\[
\sum_{i}\ket{\xi_i}\bra{\xi_i},\quad
\xi_i\in p{\cc}^k\setminus\{0\}.
\]
Fix a nonzero $\eta\in {\cc}^k$ and set $C_i=\ket{\xi_i}\bra{\eta}\in \mk$. A direct calculation shows that
\begin{equation*}
\sum_{i}\lV \xi_i\rV^2
\sigma_{C_i} = \omega^\bb_R,
\end{equation*}
where $\sigma_{B}(A) = \lmk \varphi^{\bb}(B^*B)\rmk^{-1}
\varphi^\bb\lmk B^*e_\bb^{-1/2}R^\bb(A)e_\bb^{-1/2} B\rmk
$ defines a state on $A_{[0,\infty)}$ which belongs to
$\caE_{R}^{k,\bb}$, for any nonzero $B\in \mk$. Therefore, it suffices to show that
$\sigma_B(\tau_{a}(\Phi^{k,\bb}_{m}))=0$ for all
$0\le a$ and nonzero $B\in \mk$. But this follows from
\begin{equation*}
\sigma_B(\tau_{a}(\Phi^{k,\bb}_{m}))
=\lim_{N\to\infty}\frac{\braket{\Gamma_{N}^{k,\bb}\left(e_\bb^{-1/2} B\right)}
{\tau_{a}(\Phi^{k,\bb}_{m})\Gamma_{N}^{k,\bb} \left(e_\bb^{-1/2}B\right)}}
{\varphi^\bb(B^*B)}
\end{equation*}
and the fact that by definition of $\Phi^{k,\bb}_{m}$
and the intersection property the numerator on the right hand side
is uniformly equal to $0$ for $N\ge m+a-1$.

Next we show $\caS_{[0,\infty)}(H^{k,\bb}_{m})\subset \caE_{R}^{k,\bb}$.
Let $\omega\in \caS_{[0,\infty)}(H^{k,\bb}_{m})$.
For each $N\ge m^{k,\bb}$, let $D_N$ be the density matrix of  
the restriction of $\omega$ to ${\mathcal A}_{[0,N-1]}$, namely 
$\omega(A) = \Tr_{[0,N-1]}(D_N A)$ for any $A\in {\mathcal A}_{[0,N-1]}$.
By the condition $\omega(\tau_a(\Phi^{k,\bb}_{m}))=0$, $0\le a\le N-m$,
and the intersection property,
we have that $\mathrm{Ran}(D_N)\subset \caG_{N}^{k,\bb}$.
Therefore, there exist $X_{i,N}\in\mk$,
$i=1,\ldots, k^2$ such that
\begin{equation}\label{DN}
D_N=\sum_i\ket{\Gamma_{N}^{k,\bb}(X_{i,N})}\bra{\Gamma_{N}^{k,\bb}(X_{i,N})}.
\end{equation}
Note that, using first Lemma \ref{lemma:injectivity} and then~(\ref{bieq}),
\begin{equation*}
1 = \sum_i \Vert \Gamma_{N}^{k,\bb}(X_{i,N}) \Vert^2 \geq \sum_i \braket{X_{i,N}}{X_{i,N}}_\bb(1-E^{k,\bb}(N))\geq \frac 12 (a^{\bb} c^\bb)^{-1}
\sum_i \Tr_{\mk} \lmk X_{i,N}^*X_{i,N}\rmk,
\end{equation*}
for $N\geq L^{k,\bbB}$. Hence, by compactness, there is a subsequence $\{N_m\}_m$ such that
\begin{equation*}
\lim_{m\to\infty}X_{i,N_m} = X_{i,\infty}
\end{equation*}
for all $i$. 
From Lemma \ref{lemma:injectivity}, we have
\begin{align}\label{xxi}
\limsup_{m\to\infty}\sum_i \Vert \Gamma_{N_m}^{k,\bb}(X_{i,N_m}-X_{i,\infty}) \Vert^2 
\le\limsup_{m\to\infty}\sum_i
(1+E^{k,\bb}(N_m))\braket{X_{i,N_m}-X_{i,\infty}}{X_{i,N_m}-X_{i,\infty}}_\bb=0.
\end{align}
From this we have
\begin{align}\label{xnom}
1 = \lim_{m\to\infty}\sum_i \Vert \Gamma_{N_m}^{k,\bb}(X_{i,N_m}) \Vert^2 
=\lim_{m\to\infty}\sum_i \Vert \Gamma_{N_m}^{k,\bb}(X_{i,\infty}) \Vert^2 =
\sum_i\braket{X_{i,\infty}}{X_{i,\infty}}_\bb.
\end{align}
Therefore, there exists a nonzero $X_{i,\infty}$. 

Now, set $Y_i:=e_\bb^{1/2}X_{i,\infty}$. Note that $Y_i\neq 0$, if $X_{i,\infty}\neq 0$. The form $\tilde\omega$ on ${\mathcal A}_{[0,\infty)}$ defined by
\begin{equation*}
\tilde\omega:= \sum_{i:X_{i,\infty}\neq 0}\braket{X_{i,\infty}}{X_{i,\infty}}_\bb\sigma_{Y_{i}}
\end{equation*}
is a state by (\ref{xnom}).
Furthermore, $\tilde\omega\in \caE_{R}^{k,\bb}$, as each
$\sigma_{Y_{i}}$ is, and we have
\begin{align*}
\tilde\omega(A)=
\sum_{i:X_{i,\infty}\neq 0}
\varphi^\bb\lmk Y_i^*e_\bb^{-1/2}R^\bb(A)e_\bb^{-1/2} Y_i\rmk
=\sum_{i:X_{i,\infty}\neq 0}
\varphi^\bb\lmk X_{i,\infty }^* R^\bb(A) X_{i,\infty}\rmk,\qquad
A\in {\caA}_{[0,\infty)}.
\end{align*}
We claim $\omega=\tilde\omega$:
For $A\in{\caA}^{\rm loc}$, from (\ref{DN}),  (\ref{xxi}) and Lemma \ref{rba},
\begin{align*}
\lv \omega(A)-\tilde\omega(A)\rv
& =\lim_{m\to\infty}\lv
\sum_i\lmk \braket{\Gamma_{N_m}^{k,\bb}(X_{i,N_m})}{A \Gamma_{N_m}^{k,\bb}
(X_{i,N_m})}
-\varphi^\bb\lmk X_{i,\infty}^*R^\bb(A)X_{i,\infty}\rmk\rmk
\rv\\
&=\lim_{m\to\infty}\lv
\sum_i\lmk\braket{
\Gamma_{N_m}^{k,\bb}(X_{i,\infty})}
{A \Gamma_{N_m}^{k,\bb}
(X_{i,\infty})}
-\varphi^\bb\lmk X_{i,\infty}^*R^\bb(A)X_{i,\infty}\rmk\rmk
\rv
=0,
\end{align*}
 and in fact $\tilde\omega=\omega$. Hence, $\omega\in \caE_{R}^{k,\bb}$. The case $\caS_{(-\infty,-1]}(H^{k,\bb}_{m}) = \caE^{k,\bb}_{L}$ is treated similarly.
\end{proof}

\begin{proofof}[Proposition~\ref{thm:phases}]
The intersection property, part (i), follows from Lemma~\ref{lem:intersection}(b). Part(iia) is a consequence of Proposition~\ref{Prop:Gapped pq}. (iib) and (iic) were the contents of Proposition~\ref{Prop:edge spaces}.
\end{proofof}

\section{A continuous path of Hamiltonians}\label{sec:PathHamiltonians}

We shall write $H\simeq_C H'$ if the translation invariant Hamiltonians $H$ and $H'$ are $C^1$-equivalent.
For $m\in\nan$ and translation invariant Hamiltonians
$H_m,H_m'$ with interaction length less than or equal to $m$,
we further write $H_m\simeq_{C,m}H_m'$
if $H_m$ and $H_m'$ are $C^1$-equivalent and the $C^1$-path 
can be taken as a path in $\caJ_m$.

In this section, we prove the main result of this paper, Theorem~\ref{mainclass}. For now, we shall use the following technical result, which will be proved in Section~\ref{sec:PathMatrices}, Proposition~\ref{skc}: If $\bb,\bb'\in B_{n,k}$ and for any $m\ge 2k(k-1)+3$, there exists a continuous map $\bar{\bbA}:[0,1]\to \Mat_{1,n}(\Mat_{k}(\cc))$, piecewise of class $C^1$, such that $\bar{\bbA}(0)=\bb$, $\bar{\bbA}(1)=\bb'$ and $\bar{\bbA}(t)\in X_{n,k,m}$ (see~(\ref{Xnkm}) for the definition of this set) with $\bar A_1(t)$ invertible for $t\in(0,1)$.

With this, the proof of $C^1$-equivalence relies on two results. Firstly, that given $\bbB\in B_{n,k}$ and two interactions of ranges $m,m'$ such that $\caG_m^{k,\bbB}$ and $\caG_{m'}^{k,\bbB}$ satisfy the intersection property, then the two corresponding Hamiltonians are equivalent: We refer to this as equivalence under changing of the interaction length, Lemma~\ref{ci}. Secondly, if $\bbB,\bbB'\in B_{n,k}$ and $H_m^{k,\bbB}, H_m^{k,\bbB'}$ are the corresponding Hamiltonians (with the same $m$ and $k$), then the smooth path of matrices $\bar{\bbA}(t)$ mentioned above yields the $C^1$-equivalence of the Hamiltonians: This is the equivalence under smooth deformations of $\bbB$, Lemma~\ref{lem:B(t)}.

We start with the equivalence under deformations of $\bbB$.
\begin{lem}\label{lem:B(t)}
Let $k\in\nan$ and 
$\bb,\bb'\in B_{n,k}$.
Then for any $m\ge k^4+1$, 
\[
H^{k,\bb}_{m}\simeq_{C,m} H^{k,\bb'}_{m}.
\]
\end{lem}
\proof 
Let $\bar{\bbA}(t)$ be the path given by Proposition~\ref{skc}. By Lemma \ref{lem:intersection}, $r_{\bar\bbA(t)}>0$ for $t\in(0,1)$ and
$(r_{\bar\bbA(t)})^{-\frac 12}\bar\bbA(t)\in B_{n,k}$.
That $r_{\bar\bbA(t)}>0$ and $\widehat \bbE^{\bar \bbA(t)}\in \caT_k$ for $t=0,1$ follows by definition.

Set $\bbA(t):=r_{\bar\bbA(t)}^{-\frac 12}\bar\bbA(t)\in B_{n,k}$, for $t\in[0,1]$. Applying Lemma \ref{lem: continuity and decay} to a continuous piecewise $C^1$-path
$[0,1]\ni t\mapsto \widehat \bbE^{\bar \bbA(t)}\in \caT_k$, we see that 
$[0,1]\ni t\mapsto r_{\bar\bbA(t)}$ is continuous and piecewise $C^1$.
Therefore, the path $[0,1]\ni t\mapsto \widehat \bbE^{ \bbA(t)}\in \caT_k$
is continuous and piecewise $C^1$.

Let $m_0:=k^4+1$. For any $m\geq m_0$, Lemma~\ref{lem:intersection} ensure that $\Gamma_m^{k,\bbA(t)}$ is injective and that $(k,\bbB)$ satisfies Property~(I,$m$). Recall the definitions at the beginning of Section~\ref{hamgap}. We claim
\begin{enumerate}
\item $l_0:=\sup_{t\in[0,1]}\bar l^{k,{\bbA}(t)}<\infty$,
\item for any $m_0\le m$, the map
\[
[0,1]\ni t\mapsto G_{m}^{k,\bbA(t)}
\]
is continuous and piecewise $C^1$,
\item for all $l,m\in\nan$ with $m_0\le m\le l$, 
\[
\gamma:=\inf_{t\in[0,1]}\gamma_{l,m}^{k,{\bbA}(t)}>0.
\]
\end{enumerate}
(i)
By Lemma \ref{lem: Inverses}, $a:=\sup_{t\in [0,1]}a^{\bbA(t)}$ and $c:=\sup_{t\in [0,1]}c^{\bbA(t)}$
are finite.
Furthermore, by Lemma \ref{lem: continuity and decay},
there exist $0<\lambda<1$ and $C>0$ such that
\begin{equation*}
\sup_{t\in[0,1]}\lV \lmk \widehat\bbE^{\bbA(t)}\rmk^l\lmk 1-P_{\{1\}}^{\widehat\bbE^{\bbA(t)}}\rmk\rV \le C \lambda^l,\qquad
l\in\nan.
\end{equation*}
Also, by Lemma \ref{lem: continuity and decay}, $[0,1]\ni t\mapsto e_{\bbA(t)}$, 
is continuous and we have $b:=\sup_{t\in[0,1]}\Tr_{\mk}(e_{\bbA(t)})<\infty$.
From these estimates and the definition of $E^{k,\bbA(t)}(N)$, $F^{k,\bbA(t)}$, we obtain the uniform bound
\begin{align*}
&\sup_{t\in[0,1]}\sup_{N\ge l}\lmk 
\sqrt{N+1}\lmk
3E^{k,\bbA(t)}(N)F^{k,\bbA(t)}+2
\rmk
\lmk
E^{k,\bbA(t)}(N)F^{k,\bbA(t)}
\rmk
+ E^{k,\bbA(t)}(N)
\rmk\\
&\le
4kC\lmk kacC+c+ab\rmk
\lmk 12kC\lmk kacC+c+ab\rmk+2\rmk
\sup_{N\ge l}\sqrt{N+1}\lambda^N
+kac C \sup_{N\ge l}\lambda^N,
\end{align*}
for all $l\in\nan$.
As $0<\lambda<1$, the right hand side converges to $0$ as $l\to\infty$.
In particular, we have $l_0:=\sup_{t\in[0,1]}\bar l^{k,{\bbA}(t)}<\infty$.
\\
(ii) 
Let $\{e_{i,j}\}_{i,j=1,\ldots,k}$ be the set of matrix units of $\mk$.
Then, for each $m\ge m_0$ and $t\in[0,1]$, 
$ G_{m}^{k,\bbA(t)}$ is the orthogonal projection onto a subspace of 
$\otimes_{i=0}^{m-1} \cc^n$ spanned by the vectors
$\{\Gamma^{k,\bbA(t)}_{m}\big(e_{ij}\big),i,j=1,\ldots,k\}$.
Injectivity of $\Gamma^{k,\bbA(t)}_{m}$ for $m\ge m_0$ means that the dimension of 
$ G_{m}^{k,\bbA(t)}$ is constant and equal to $k^2$, for $t\in[0,1]$.
Hence, from Lemma \ref{pob}, (ii) holds.
\\
(iii) For all $l,m\in\nan$ with $m_0\le m\le l$, 
we have $l\ge m\ge m_0\ge m^{k,\bbA(t)}$ for $t\in [0,1]$.
Therefore, by Proposition \ref{Prop:Gapped pq},
the lowest eigenvalue of 
$X(t)=
\lmk H^{k,\bbA(t)}_{m}\rmk_{[0,l-1]}
$
is $0$ and the corresponding spectral projection is $G^{k,\bbA(t)}_{l}$.
Therefore, the path $X:[0,1]\mapsto X(t)$ is a continuous and piecewise $C^1$-path of constant rank positive matrices.
(iii) now follows from Lemma~\ref{poa} applied to this path.

Fix $m\ge m_0$ and $l>\max\{l_0,m\}$, where the $\max$ is finite by Claim~(i) above.
Applying Proposition \ref{Prop:Gapped pq}, 
for $N\in\nan$ with $N\ge m\ge m_0\ge  m^{k,\bbA(t)} $,
the lowest eigenvalue of 
$
(H^{k,\bbA(t)}_{m})_{[0,N-1]}
$
is $0$ and the corresponding spectral projection is $G^{k,\bbA(t)}_{N}$.
 Furthermore, 
 \begin{align*}
 \frac{\gamma}{4(l+2)}\lmk 1-G^{k,\bbA(t)}_{N}\rmk\le \frac{\gamma_{l,m}^{k,\bbA(t)}}{4(l+2)}\lmk 1-G^{k,\bbA(t)}_{N}\rmk
\le
\ (H^{k,\bbA(t)}_{m})_{[0,N-1]}
 \end{align*}
for all $N>l$, because
$\max\{\bar l^{k,\bbA(t)},m\}\le \max\{l_0,m\}< l$, where the first inequality is Claim~(iii) above. 
Hence, for $m\ge m_0$, the path of positive interactions 
$[0,1]\ni t\mapsto \Phi^{k,\bbA(t)}_{m}\in\caJ_m$ is continuous and piecewise $C^1$ by Claim~(ii),
 $\Phi^{k,\bbA(0)}_{m}=\Phi^{k,\bb}_{m}$, $\Phi^{k,\bbA(1)}_{m}=\Phi^{k,\bb'}_{m}$, and the Hamiltonian associated with $\Phi^{k,\bbA(t)}_{m}$ has a gap 
$\frac{\gamma_{l,m}^{k,\bbA(t)}}{4(l+2)}$ for each $t\in[0,1]$, with the uniform lower bound $\inf_{t\in[0,1]}\frac{\gamma_{l,m}^{k,\bbA(t)}}{4(l+2)}\ge  \frac{\gamma}{4(l+2)}>0$. This proves
$H^{k,\bb}_{m}\simeq_{C,m} H^{k,\bb'}_{m}$.
\endproof

We now turn to the problem of changing the interaction length.
\begin{lem}\label{ci}
Let $k\in\nan$
and $\bb\in B_{n,k}$.
For any $m,m'\ge m^{k,\bb}$,
\[
H^{k,\bb}_{m}\simeq_{C,\max\{m,m'\}} H^{k,\bb}_{m'}.
\]
\end{lem}
\proof
We may assume $m\neq m'$.
For each $t\in[0,1]$ set
\begin{align}
\Phi(t;X):=
\begin{cases}
(1-t)\tau_x\lmk 1-G^{k,\bb}_{m}\rmk,& \text{if }   X=[x,x+m-1]   \text{ for some }   x\in\bbZ\\
t\tau_x\lmk 1-G^{k,\bb}_{m'}\rmk,& \text{if }   X=[x,x+m'-1]   \text{ for some }   x\in\bbZ\\
0,&\text{otherwise}
\end{cases}\,,
\end{align}
This defines a $C^1$-path $\Phi:[0,1]\to \caJ$
such that $\Phi(0)=\Phi^{k,\bb}_{m}$ and $\Phi(1)=\Phi^{k,\bb}_{m'}$.
The interaction length of $\Phi(t)$ is less than or equal to $\max\{m,m'\}$.
Let $H(t)$ be the Hamiltonian associated with the interaction $\Phi(t)$.
For each $N\ge \max\{m,m'\}$ and $t\in(0,1)$,
the kernel of $\lmk H(t)\rmk_{[0,N-1]}$ is given by
\[
\ker \lmk H(t)\rmk_{[0,N-1]}=
\ker\lmk H^{k,\bb}_{m}\rmk_{[0,N-1]}\cap\ker \lmk H^{k,\bb}_{m'}\rmk_{[0,N-1]}={\mathcal G}^{k,\bb}_{N}=
 \ker \lmk H^{k,\bb}_{m}\rmk_{[0,N-1]}=\ker\lmk  H^{k,\bb}_{m'}\rmk_{[0,N-1]},
\]
by Proposition \ref{Prop:Gapped pq}.
Therefore, the Hamiltonian $H(t)$ has a spectral gap, namely,
for any $l> \max\{\bar l^{k,\bb},m,m'\}$ and $N\ge l+1$,
\[
\frac{1}{4(l+2)}\lmk(1-t)\gamma_{l,m}^{k,\bb}+
t\gamma_{l,m'}^{k,\bb}
\rmk(1-G^{k,\bb}_{N, p,q})\le
(1-t)(H^{k,\bb}_{m})_{[0,N-1]}+t(H^{k,\bb}_{m'})_{[0,N-1]}
=\lmk H(t)\rmk_{[0,N-1]},
\] for all $t\in[0,1]$.
Hence we have $H^{k,\bb}_{m}\simeq_{C,\max\{m,m'\}} H^{k,\bb}_{m'}$.
\endproof

\begin{proofof}[Theorem~\ref{mainclass}]
Let $k\neq k'$ and assume by contradiction that $H\simeq_C H'$. Then Theorem~\ref{qla} yields a bijective map between ground state spaces of $H$ and $H'$, in the bulk and at the edges. But Proposition~\ref{thm:phases}(c) implies that the edge spaces of $H$ and $H'$ are of different dimensions, a contradiction. 

Reciprocally, assume that $k = k'$. Let $H^{k,\bb}_{m},H^{k,\bb'}_{m'}\in {\mathcal H}_k$, where $m\geq m^{k,\bbB},m'\geq m^{k,\bbB'}$.
Let $M:=\max\{m,k^4+1\}$, $M':=\max\{m',k^{4}+1\}$,
and set $\tilde m=\max\{M,M'\} = \max\{m,m', k^4+1\}$. By Lemma \ref{ci}, we have
\begin{equation*}
H^{k,\bb}_{m} \simeq_{C,\tilde m} H^{k,\bb}_{\tilde m},\qquad H^{k,\bb'}_{m'} \simeq_{C,\tilde m} H^{k,\bb'}_{\tilde m}.
\end{equation*}
Furthermore, Lemma~\ref{lem:B(t)} yields the equivalence
\begin{equation*}
H^{k,\bb}_{\tilde m}\simeq_{C,\tilde m} H^{k,\bb'}_{\tilde m}
\end{equation*}
and the theorem follows by transitivity of $\simeq_{C,\tilde m}$.
\end{proofof}

\section{Piecewise $C^1$-paths of matrices}\label{sec:PathMatrices}
Recall the definitions of $\caK_m({\bb})$ and $X_{n,k,m}$ introduced in Section \ref{edge}
\begin{prop}\label{skc}
Let $n,k,m\in{\nan}$ such that $2k(k-1)+3\le m$.
Then for any $\bbA,\bbE\in \Mat_{1,n}(\Mat_{k}(\cc))$, there exists a continuous map 
$\bb:[0,1]\to \Mat_{1,n}(\Mat_{k}(\cc))$, piecewise of class $C^1$, such that
$\bb(0)=\aaa$, $\bb(1)=\ee$ and $\bb(t)\in X_{n,k,m}$ with invertible $B_1(t)$ for $t\in(0,1)$.
\end{prop}

Here, we give a constructive proof, rather than
just showing the existence of the path.
The strategy to prove this is to consider simple subsets of $X_{n,k,m}$ that can be constructively proven to be arcwise connected. 
Clearly, if $k=1$, then ${\mathcal K}_m(\bb)=\mk$ for any $m\in\nan$ and nonzero $\bb\in \Mat_{1,n}(\Mat_{k}(\cc))$.
Therefore, we may assume that $2\le k$. Throughout the proof, we fix an orthonormal basis $\{e_{\alpha}\}_{\alpha=1}^k$ of $\mk$.

For $k\in \nan$, let ${\mathcal P}_{k}:=\{(i,j)\in \{1,\ldots,k\}\times
\{1,\ldots, k\}\mid i\neq j \}$ and
\begin{align}\label{sk}
S_k:=
\left\{
\begin{gathered}
{\mathbb \lambda}=(\lambda_1,\ldots,\lambda_k)\in
{(\cc\setminus\{0\})}^k\left \vert
\begin{gathered}
\lambda_i\neq\lambda_j,\; {\rm if}\;(i, j)\in{\mathcal P}_k,\\
\frac{\lambda_i}{\lambda_j}
\neq\frac{\lambda_{i'}}{\lambda_{j'}},\;
{\rm if }\; (i,j)\neq(i',j'),\;
 (i, j),(i',j')\in{\mathcal P}_k
\end{gathered}\right.
\end{gathered}
\right\}.
\end{align}

For $n,k\in\nan$, let
\begin{align*}
Y_{n,k}:= \left\{
\bb\in \Mat_{1,n}(\Mat_{k}(\cc))
\left\vert
\begin{gathered}
B_1=\sum_{\alpha=1}^{k}\lambda_{\alpha}
\ket{e_{\alpha}}\bra{e_{\alpha}},\;\text{where}\;
\lambda \in S_k,\\
\text{and}\;\langle B_2 e_{\alpha}, e_{\beta}\rangle\neq 0,\quad
\alpha,\beta=1,\ldots,k
\end{gathered}\right.
\right\}.
\end{align*}
Furthermore, for $\bb=(B_1,\ldots,B_n)\in \Mat_{1,n}(\Mat_{k}(\cc))$
and an $R\in GL(k,\bbC)$, we denote 
\[
R\bb R^{-1}:
=(RB_1R^{-1},\ldots,RB_nR^{-1})\in \Mat_{1,n}(\Mat_{k}(\cc)).
\]
For $\bb,\bb'\in \Mat_{1,n}(\Mat_{k}(\cc))$, we say that $\bb$ is similar to 
$\bb'$ if there exists an $R\in GL(k,\bbC)$ such that
$R\bb R^{-1}=\bb'$.
Define
\begin{equation*}
Z_{n,k} := \left\{\bb\in \Mat_{1,n}(\Mat_{k}(\cc))\mid
\bb  
\text{ is similar to an element in }Y_{n,k}\right\}.
\end{equation*}

Note that
$R\in GL(k,\cc)$ which diagonalizes $B_1$ is not unique.
However, for any $\bb\in Z_{n,k}$ and invertible
$P\in \Mat_k(\cc)$ such that $PB_1 P^{-1}$ is diagonal with respect to
the basis $\{e_{\alpha}\}_{\alpha=1}^k$,
we have $P^{-1}\bb P\in Y_{n,k}$.
This follows from the condition $\lambda_i\neq \lambda_j$ for $(i,j)\in {\mathcal P}_k$ in the definition of $S_k$.

The proof of Proposition~\ref{skc} will now proceed through a series of lemmas.
\begin{lem}
Let $2\le n,k\in{\nan}$ and $2k(k-1)+3\le m\in\nan$.
Then $Y_{n,k}\subset X_{n,k,m}$.
\end{lem}
\begin{rem}\label{bnkne}
In particular, $X_{n,k,m}$ is nonempty for $2k(k-1)+3\le m\in\nan$.
From Lemma \ref{lem:intersection}, this means $B_{n,k}$ is not empty and contains an element $\bb$ with an invertible
$B_1$.
\end{rem}

\begin{proof}
Let $\bb=(B_1,\ldots,B_n)\in Y_{n,k}$.
We first claim that for each $(a,b)\in{\mathcal P}_k$,
there exists a nonzero vector 
$\zeta^{(a,b)}=(\zeta_{l}^{(a,b)})_{l=0,\cdots,{k(k-1)}}
\in \cc^{{k(k-1)}+1}$ such that
\begin{equation}\label{orthogonality}
\sum_{l=0}^{{k(k-1)}} \zeta_{l}^{(a,b)}
\lmk
\frac{\lambda_{\alpha}}{\lambda_{\beta}}\rmk^l
=\delta_{\alpha,a}\delta_{\beta,b},
\end{equation} 
for all $\alpha,\beta=1,\ldots,k$, where $\lambda\in S_k$ are the eigenvalues of $B_1$. To do this, we define for each $\alpha,\beta=1,\ldots,k$
\begin{align*}
v_{\alpha,\beta}:=\left(
\begin{array}{c}
1\\
\lmk
\frac{\lambda_{\alpha}}{\lambda_{\beta}}\rmk^1\\
\lmk
\frac{\lambda_{\alpha}}{\lambda_{\beta}}\rmk^2\\
\vdots\\
\lmk
\frac{\lambda_{\alpha}}{\lambda_{\beta}}\rmk^{{k(k-1)}}
\end{array}
\right)\in \cc^{{k(k-1)}+1}
\end{align*}
The condition $\lambda\in S_k$ implies that the determinant of the
following Vandermonde matrix
\begin{align*}
\left(
\begin{array}{ccccccccccccc}
v_{1,1}&v_{1,2}
&v_{1,3}&\cdots&v_{1,k}
&v_{2,1}&v_{2,3}&\cdots&v_{2,k}&
\cdots&v_{k,1}&\cdots&v_{k,k-1}
\end{array}
\right)\in\Mat_{{k(k-1)}+1}(\cc)
\end{align*}
is nonzero.
This means the set of vectors
$\{v_{i,j}\}_{(i,j)\in{\mathcal P}_k}\cup\{v_{1,1}\}$
are linearly independent.
Therefore, for each $(a,b)\in{\mathcal P}_k$, 
there exists 
a nonzero vector 
$\zeta^{(a,b)}$
such that
\[
\zeta^{(a,b)}\perp 
\{v_{i,j}\}_{(i,j)\in{\mathcal P}_k, (i,j)\neq(a,b)}
\cup\{v_{11}\},
\]
and 
\[
\left\langle\zeta^{(a,b)},v_{a,b}\right\rangle=1.
\]
Hence we have shown the claim.

With $\bbB \in Y_{n,k}$, equation~(\ref{orthogonality}) and a short calculation yield that
\begin{equation*}
\sum_{l=0}^{{k(k-1)}} \zeta_{l}^{(a,b)} B_1^l B_2 B_1^{{k(k-1)}-l}
= \lambda_{b}^{{k(k-1)}}
\langle{e_a},B_2 e_{b}\rangle
\ket{e_{a}}\bra{e_{b}},
\end{equation*}
for each $(a,b)\in{\mathcal P}_k$. 
As $\lambda_{b}^{{k(k-1)}}
\langle{e_a},B_2 e_{b}\rangle
\neq 0$, this means
$\ket{e_{a}}\bra{e_{b}}\in {\mathcal K}_{k(k-1)+1}(\bb)$
for any $(a,b)\in{\mathcal P}_k$.

Finally, for any $(a,b)$, possibly $ a= b$, we choose
 $a',b'=1,\ldots,k$ with $a\neq a',b\neq b'$, so that $\ket{e_{a}}\bra{e_{a'}}, \ket{e_{b'}}\bra{e_{b}}\in{\mathcal K}_{k(k-1)+1}(\bb)$. Hence,
\begin{equation*}
{\mathcal K}_{m}
(\bb)
\ni
\ket{e_{a}}\bra{e_{a'}}
B_2 B_1^{m-2k(k-1)-3}\ket{e_{b'}}\bra{e_{b}}
=\lambda_{b'}^{m-2k(k-1)-3}\langle e_{a'},B_2 e_{b'}\rangle
\ket{e_{a}}\bra{e_{b}}.
\end{equation*}
As $\lambda_{b'}^{m-{2}k(k-1)-3}\langle e_{a'},B_2 e_{b'}\rangle\neq 0$,
this means
$\ket{e_{a}}\bra{e_{b}}\in {\mathcal K}_{m}
(\bb)$ for each $a,b = 1,\ldots k$,
and we conclude ${\mathcal K}_{m}
(\bb)=\Mat_k(\cc)$.
Thus we obtain
$Y_{n,k}\subset X_{n,k,m}$.
\end{proof}
From this, we have $Z_{n,k}\subset X_{n,k,m}$
for $2k(k-1)+3\le m$.
Next, we show that $Z_{n,k}$ is arcwise connected.
\begin{lem}\label{zac}
For $n,k\in\nan$ with $n,k\ge 2$, and $\aaa,\ee\in Z_{n,k}$,
there exists a $C^{\infty}$-path
$\bb:[0,1]\to Z_{n,k}$
such that $\bb(0)=\aaa$, $\bb(1)=\ee$.
\end{lem}
\begin{proof}
By definition, if $\bbA, \bbE \in Z_{n,k}$, there exist $P_\bbA,P_\bbE\in GL(k,\bbC)$ such that $\overline{\bbA}
=(\bar A_1,\ldots,\bar A_n):=P_\bbA \bbA P_\bbA^{-1}\in Y_{n,k}$ and $\overline{\bbE}=(\bar E_1,\ldots,\bar E_n):=P_\bbE \bbE P_\bbE^{-1}\in Y_{n,k}$. As $GL(k,\bbC)$ is connected, there exists a $C^{\infty}$-path $P : [0,1]\to GL(k,\bbC)$ such that $P(0)=P_\bbA$ and $P(1)=P_\bbE$.

By assumption, there exist $\lambda=(\lambda_1,\ldots,\lambda_n), \mu=(\mu_1,\ldots,\mu_n)
\in S_k$ such that
\begin{equation*}
\overline{A}_1=\sum_{\alpha=1}^{k}\lambda_{\alpha} \ket{e_{\alpha}}\bra{e_{\alpha}},\quad
\overline{E}_1=\sum_{\alpha=1}^{k}\mu_{\alpha} \ket{e_{\alpha}}\bra{e_{\alpha}}.
\end{equation*}
By Lemma~\ref{skp}, there is a $C^{\infty}$-path $\lambda:[0,1]\to S_k$ such that $\lambda(0)=\lambda$, and $\lambda(1)=\mu$. Furthermore, let $\xi_{\alpha,\beta}=\langle e_{\alpha}, \overline{A}_2 e_{\beta}\rangle$ and $\chi_{\alpha,\beta}=\langle e_{\alpha}, \overline{E}_2 e_{\beta}\rangle$. By assumption again, $\xi_{\alpha,\beta},\chi_{\alpha,\beta}\neq0$. Then Lemma~\ref{sf} yields a $C^{\infty}$-path $\xi_{\alpha,\beta}:[0,1]\to \cc\setminus\{0\}$ such that $\xi_{\alpha,\beta}(0) = \xi_{\alpha,\beta}$ and $\xi_{\alpha,\beta}(1) = \chi_{\alpha,\beta}$. Now, we define for $t\in[0,1]$
\begin{align*}
\overline{A}_1(t) &= \sum_{\alpha=1}^{k}\lambda_{\alpha}(t) \ket{e_{\alpha}}\bra{e_{\alpha}}, \\
\overline{A}_2(t) &= \sum_{\alpha,\beta=1}^{k} \xi_{\alpha,\beta}(t) \ket{e_{\alpha}}\bra{e_{\beta}}, \\
\overline{A}_i(t) &= (1-t) \overline{A}_i + t\overline{E}_i,\qquad 3\leq i\leq n.
\end{align*}
Clearly, $\overline{\bbA}(t) = (\overline{A}_1(t),\ldots,\overline{A}_n(t))\in Y_{n,k}$. Finally, the path $Z_{n,k}\ni\bbB(t) = P(t)^{-1}\overline{\bbA}(t)P(t)$ 
is $C^{\infty} and $ connects $\bbA$ to $\bbE$, which concludes the proof.
\end{proof}
Now we connect an arbitrary element in 
$\Mat_{1,n}(\Mat_{k}(\cc))$ with 
an element in $Z_{n,k}$.

\begin{lem}\label{gtz}
Let $n,k\in\nan$ with $n,k\ge 2$.
For any $\aaa\in \Mat_{1,n}(\Mat_{k}(\cc))$, there exists a $C^{\infty}$-
path $\bb : [0,1]\to \Mat_{1,n}(\Mat_{k}(\cc))$
such that $\bb(0)=\aaa$
and $\bb(t)\in Z_{n,k}$
for all $t\in (0,1]$.
\end{lem}
\begin{proof}
Let $\aaa=(A_1,\ldots,A_n)\in \Mat_{1,n}(\Mat_{k}(\cc))$.
We consider the Jordan normal form of $A_1$ with respect to the
orthonormal basis $\{e_{\alpha}\}_{\alpha=1}^k$.
Let $n_1,\ldots,n_M\in\nan$ be the dimension of each Jordan cell
of $A_1$, so that $\sum_{l=1}^Mn_l=k$.
For $1\leq l \leq M$, denote by 
$J_l$ the $l$-th
Jordan cell with eigenvalue $\lambda_l$.
We further group the orthonormal basis 
$\{e_{\alpha}\}_{\alpha=1}^k$  
corresponding to the decomposition,
$\cc^k=\cc^{n_1}\oplus\cdots\oplus \cc^{n_M}$ and label them
$\{f_{\alpha}^{(l)}\}_{\alpha=1}^{n_l}$, $l=1,\ldots,M$.
For each $l$, $\{f_{\alpha}^{(l)}\}_{\alpha=1}^{n_l}$ is an orthonormal basis of $\cc^{n_l}$.
With these notations, each $J_l$ can be written
\[
J_l=
\sum_{\alpha=1}^{n_l}\lambda_{l}
\ket{f_{\alpha}^{(l)}}\bra{f_{\alpha}^{(l)}}
+\sum_{\alpha=2}^{n_l}
\ket{f_{\alpha-1}^{(l)}}\bra{f_{\alpha}^{(l)}}.
\]
The Jordan normal form now reads $A_1=R J R^{-1}$ for a $R\in GL(k,\bbC)$, and where $J:=J_1\oplus\cdots\oplus
J_M$. It will also be useful to gather the eigenvalues with their multiplicities: for $l=1,\ldots,M$ and $\alpha=1,\ldots,n_l$, we define $\lambda_{\alpha}^{(l)}:=\lambda_l$ and let $\lambda=(\lambda_{1}^{(1)},\ldots,\lambda_{n_1}^{(1)},\lambda_1^{(2)},\ldots,\lambda_{n_{M-1}}^{(M-1)},\lambda_{1}^{(M)},\ldots,\lambda_{n_M}^{(M)})\in \bbC^k$.

For each $l=1,\ldots,M$ and 
$\alpha=1,\ldots,n_l$, we set $m_{(l,\alpha)}:=\alpha+\sum_{i=1}^{l-1}
n_{i}$.
Let $N_{(l,\alpha)}:=2^{m_{(l,\alpha)}+1}$,
and
$\lambda_{\alpha}^{(l)}(t):=\lambda_{\alpha}^{(l)}+t^{N_{(l,\alpha)}}$
for $t\ge 0$.
Corresponding to the decomposition
$\cc^k=\cc^{n_1}\oplus\cdots\oplus \cc^{n_M}$,
we define
$\lambda(t)=:(\lambda_{1}^{(1)}(t),\ldots,
\lambda_{n_1}^{(1)}(t),\lambda_1^{(2)}(t),
\ldots,\lambda_{n_{M-1}}^{(M-1)}(t)
,\lambda_{1}^{(M)}(t),\ldots,\lambda_{n_M}^{(M)}(t))
\in\cc^k$ and note that $\lambda(0) = \lambda$. By Lemma \ref{pol}, there exists $1>\delta_1>0$ such that
$\lambda(t)\in S_k$ for $t\in(0,\delta_1)$. Now, the following matrix
\[
J(t):=\sum_{l=1}^M\sum_{\alpha=1}^{n_l}
\lambda_{\alpha}^{(l)}(t)
\ket{f_{\alpha}^{(l)}}
\bra{f_{\alpha}^{(l)}}
+\sum_{l=1}^M\sum_{\alpha=2}^{n_l}
\ket{f_{\alpha-1}^{(l)}}\bra{f_{\alpha}^{(l)}},
\]
satisfies the assumptions of Lemma~\ref{jd} for each $t\in(0,\delta_1)$. 
Define a diagonal matrix $D(t)=D_1(t)\oplus\cdots \oplus D_M(t)$ and an invertible matrix
$P(t)=P_1(t)\oplus\cdots\oplus P_M(t)\in GL(k,\bbC)$ such that
\begin{align*}
D_l(t)=\sum_{\alpha=1}^{n_l}
\lambda_{\alpha}^{(l)}(t)\ket{f_{\alpha}^{(l)}}
\bra{f_{\alpha}^{(l)}},\quad
P_l(t)=\sum_{\alpha,\beta=1}^{n_l}
c_{\beta\alpha}^{(l)}(t)\ket{f_{\beta}^{(l)}}\bra{f_{\alpha}^{(l)}},\quad
P_l(t)^{-1}=\sum_{\alpha,\beta=1}^{n_l}
d_{\alpha\beta}^{(l)}(t)\ket{f_{\alpha}^{(l)}}\bra{f_{\beta}^{(l)}}
\end{align*}
with
\begin{align*}
c_{\beta\alpha}^{(l)}(t)
=\begin{cases}
\prod_{j=\beta}^{\alpha-1}\frac{1}{\lambda_{\alpha}^{(l)}(t)- \lambda_j^{(l)}(t)} &\alpha>\beta \\
1& \alpha=\beta\\
0& \alpha<\beta
\end{cases}
\,,\qquad
d_{\alpha\beta}^{(l)}
=
\begin{cases}
\prod_{j=\alpha+1}^{\beta}\frac{1}{\lambda_{\alpha}^{(l)}(t)-\lambda_j^{(l)}(t)}& \alpha<\beta \\
1& \alpha=\beta\\
0& \alpha>\beta
\end{cases}\,.
\end{align*}
Then, by Lemma~\ref{jd},
we have $J(t)=P(t)D(t)P(t)^{-1}$. 
As $t\mapsto \lambda(t)$ is $C^{\infty}$, $\bar B_1:[0,\delta_1)\to \Mat_k(\cc)$ defined by
\begin{equation*}
\bar B_1(t):=R J(t) R^{-1}=RP(t)D(t)(RP(t))^{-1}
\end{equation*}
is a $C^{\infty}$-path with $\bar B_1(0)=RJR^{-1}=A_1$.

From the representation of $P(t)$, 
matrix elements of $(RP(t))^{-1}A_2RP(t)$ are of the form
\begin{align}\label{elements}
\la
f_{\beta}^{(l)}, P(t)^{-1}\lmk R^{-1} A_2 R\rmk P(t) f_{\beta'}^{(l')}
\ra
=\la
f_{n_l}^{(l)},\lmk R^{-1} A_2 R\rmk  f_{1}^{(l')}
\ra
d_{\beta n_l}^{(l)}(t)c_{1 \beta'}^{(l')}(t)
+g_{\beta \beta'}^{(ll')}(t),
\end{align}
where  
\begin{align}\label{gt}
\lv
g_{\beta\beta'}^{(ll')}(t)
\rv
\le
\lV R^{-1} A_2 R\rV
\cdot
\sum_{\substack{\alpha\ge \beta,\alpha'\le \beta', \\ (\alpha,\alpha')\neq (n_l,1)}}
\lv
\frac{d_{\beta\alpha}^{(l)}(t)}{d_{\beta n_l}^{(l)}(t)}
\rv
\lv
\frac{c_{\alpha'\beta'}^{(l')}(t)}{c_{1\beta' }^{(l')}(t)}
\rv
\lv
d_{\beta n_l}^{(l)}(t)c_{1\beta'}^{(l')}(t)
\rv
\le
\lV R^{-1} A_2 R\rV
\cdot
\lv
d_{\beta n_l}^{(l)}(t)c_{1\beta'}^{(l')}(t)
\rv
\cdot k^2 t^2,
\end{align}
for $t\in(0,\delta_1)$. 
Here we used the estimate
\begin{align*}
&\lv
\frac{d_{\beta\alpha}^{(l)}(t)}{d_{\beta n_l}^{(l)}(t)}
\rv
=\prod_{j=\alpha+1}^{n_l}\lv \lambda_{\beta}^{(l)}(t)-\lambda_{j}^{(l)}(t)\rv
=\prod_{j=\alpha+1}^{n_l}\lv t^{N_{(l,\beta)}}-t^{N_{(l,j)}}\rv
\le t^2,\\
&\lv
\frac{c_{\alpha'\beta'}^{(l')}(t)}{c_{1\beta' }^{(l')}(t)}
\rv
=
\prod_{j=1}^{\alpha'-1}\lv \lambda_{\beta'}^{(l')}(t)-\lambda_{j}^{(l')}(t)\rv
= 
\prod_{j=1}^{\alpha'-1}\lv t^{N_{(l',\beta')}}-t^{N_{(l',j')}}\rv
\le t^2,
\end{align*}
if $\beta\le \alpha\le n_l-1$, 
$2\le \alpha'\le \beta'$,
for $t\in(0,\delta_1)\subset (0,1)$.

Now we define $\bar B_2(t)$ for each $t\in[0,\delta_1)$
as
\begin{align*}
\bar B_2(t)=
A_2+
tR\lmk
\sum_{l=1}^{M}\sum_{l'=1}^{M} 
\ket{f_{n_l}^{(l)}}\bra{f_1^{(l')}}
\rmk R^{-1}.
\end{align*}
This gives a $C^{\infty}$-path in $\Mat_k(\cc)$
with $\bar B_2(0)=A_2$.

We claim that there exists $\delta_1\geq \delta_2>0$ such that
the matrix elements of $(RP(t))^{-1}\bar B_2(t)RP(t)$ are nonzero
for $t\in(0,\delta_2)$. A simple computation using~(\ref{elements}) yields the bound
\begin{equation*}
\lv \la f_{\beta}^{(l)}, (RP(t))^{-1}\bar B_2(t)RP(t)  f_{\beta'}^{(l')} \ra \rv
\geq 
\bigg\vert
 \lv \la f_{n_l}^{(l)},\lmk R^{-1} A_2 R\rmk  f_{1}^{(l')}\ra + t \rv \lv d_{\beta,n_l}^{(l)}(t)c_{1,\beta'}^{(l')}(t) \rv - \lv g_{\beta,\beta'}^{(ll')}(t) \rv \bigg\vert.
\end{equation*}
where $d_{\beta,n_l}^{(l)}(t)c_{1,\beta'}^{(l')}(t)\neq 0$ by the definition
for $t\in (0,\delta_1)$. Since there is $0<\delta_2<\delta_1$
such that such that
\begin{equation*}
\lV R^{-1} A_2 R\rV \cdot k^2 t^2 < \vert \langle f_{n_l}^{(l)},\lmk R^{-1} A_2 R\rmk  f_{1}^{(l')}\rangle + t \vert,\quad
t\in (0,\delta_2)
\end{equation*}
the bound~(\ref{gt}) implies the claim.

Finally, we choose $0<\delta<\delta_2$ and define a path $\bb:[0,1]\to \Mat_{1,n}(\Mat_k(\cc))$ by
$B_1(t):=\bar B_1(\delta t)$, 
$B_2(t):=\bar B_2(\delta t)$,
and $B_i(t):=A_i$ for $i=3,\ldots,n$. By the above, $(RP(t))^{-1}\bbB(t) (RP(t))\in Y_{n,k}$ for all $t\in(0,1]$.
Hence $\bb(t)\in Z_{n,k}$.
Furthermore, $\bb$ is $C^{\infty}$ and $\bb(0)=\bbA$, and
we have obtained a path satisfying the conditions in the Lemma.
\end{proof}
With this, we are now ready to prove the main proposition of this section.
\begin{proofof}[Proposition \ref{skc}]
Recall that it suffices to consider the case $k\ge 2$.
For any $\aaa,\ee\in \Mat_{1,n}(\Mat_{k}(\cc))$, Lemma~\ref{gtz} yields two $C^{\infty}$-paths 
$\bb_{\aaa},\bb_{\ee} : [0,1]\to \Mat_{1,n}(\Mat_{k}(\cc))$
such that $\bb_{\aaa}(0)=\aaa$, $\bb_{\ee}(0)=\ee$
and $\bb_{\aaa}(t),\bb_{\ee}(t)\in Z_{n,k}\subset X_{n,k,m}$
for all $t\in (0,1]$.
By Lemma~\ref{zac}, $\bb_{\aaa}(1),\bb_{\ee}(1)\in Z_{n,k}$
can be connected by a $C^{\infty}$-path $\bb_{\mathrm{mid}}:[0,1]\to \Mat_{1,n}(\Mat_{k}(\cc))$ in
$Z_{n,k}\subset X_{n,k,m}$.
Hence, the path $\bb:[0,1]\to \Mat_{1,n}(\Mat_{k}(\cc))$ defined by
\begin{equation*}
\bbB(t) := 
\begin{cases}
\bbB_\bbA(3t) & 0\leq t\leq 1/3 \\
\bbB_{\mathrm{mid}}(3(t-1/3)) & 1/3 < t\leq 2/3 \\
\bbB_\bbE(3(1-t)) & 2/3 \leq t\leq 1 \\
\end{cases}
\end{equation*}
is a continuous path with $\bb(0)=\aaa$ and $\bb(1)=\ee$ and such that 
$\bb(t)\in Z_{n,k}\subset X_{n,k,m}$ for $t\in(0,1)$. 
It is everywhere continuously differentiable but at $t=1/3$ and $t=2/3$.
Furthermore $\bb(t)\in Z_{n,k}$ implies the invertivility of $B_1(t)$ for $t\in (0,1)$.
\end{proofof}

\section*{acknowledgments}
The authors are grateful to B.~Nachtergaele for many discussions.

S.B. acknowledges the generous support of the graduate school of mathematical sciences of the university of Tokyo and of the Institute for Theoretical Physics of the ETH Zurich where some of this work was carried out. He also wishes to thank M.~Fraas for pointing out the transversality theorem.

Y.O. is
supported by JSPS 
Grant-in-Aid for Young Scientists (B), 
and Inoue Science Research Award.
She  acknowledges the generous support of University of California, Davis where some of this work was carried out. 

\appendix
\section{Continuous paths in $S_k$}
We denote the Euclidean norm of $\rr^k$ by $\lV\;\rV_{\rr^k}$.
\begin{lem}\label{deformation}
Let $k,l\in \nan$ with $l+1<k$ and $N$ an $l$-dimensional submanifold of ${\mathbb R}^k$ without boundary.
Suppose that $\xi:[0,1]\to{\mathbb R}^k$ is  a $C^{\infty}$-path, such that $\xi(0),\xi(1)\in N^c$.
Then for any $\varepsilon>0$, there exists a $C^{\infty}$-path $\xi_{\varepsilon}:[0,1]\to \rr^k$
such that
\begin{align}
&\sup_{t\in[0,1]}\lV \xi_{\varepsilon}(t)-\xi(t)\rV_{\rr^k}<\varepsilon,\nonumber\\
&\xi_{\varepsilon}([0,1])\cap N=\emptyset,\nonumber\\
&\xi_{\varepsilon}(0)=\xi(0),\;\xi_{\varepsilon}(1)=\xi(1).
\end{align}
\end{lem}
\begin{proof}
Let $\delta=\min \{\frac{\varepsilon}{2}, \dr(\xi(0),N), \dr(\xi(1),N)\}>0$.
Define a closed subset $C$ and an open subset $U$ in $[0,1]$ by
\begin{align*}
&C:=\{t\in[0,1]\mid d_{\rr^k}(\xi(t),N)\le \frac{\delta}{3}\},\\
&U:=\{t\in[0,1]\mid d_{\rr^k}(\xi(t),N)< \frac{2\delta}{3}\}.
\end{align*}
Then $C$ is a subset of $U$ and there exists a $C^{\infty}$-path $\gamma:[0,1]\to[0,1]$
such that $\gamma\vert_C=1$ and $\gamma\vert_{U^c}=0$.

Let $S$ be an open ball in $\rr^k$ centered at the origin with radius $\delta$.
Define a $C^{\infty}$-map $F:[0,1]\times S\to \rr^k$ by
\[
F(t,s):=\xi(t)+\gamma(t)s.
\]

The map $F$ is transversal to $N$. For if $(t,s)\in F^{-1}(N)$, then
$\gamma(t)\neq 0$, by the definition of $\gamma$.
 For any $t\in[0,1]$ with $\gamma(t)\neq 0$, the map $S\rightarrow\bbR^{k}$ given by $s\mapsto\xi(t) + \gamma(t) s$ is a submersion.
Hence, $F$ is a submersion at $(t,s)$ with $\gamma(t)\neq 0$, in particular,
at $(t,s)\in F^{-1}(N)$.
Therefore, $F$ is transversal to $N$.

For the restriction $\partial F$ of $F$ to the boundary $\{0,1\}\times S$, $(\partial F)^{-1}(N)=\emptyset$.
Hence  $\partial F$ is transversal to $N$.

We conclude by the Transversality Theorem~(see e.g.~\cite{Guillemin:1974aa}) that for almost every $s\in S$, the map $\xi_s(t) := \xi(t) + \gamma(t) s$ is transversal to $N$. Since $\mathrm{dim}(N) + 1 = l+1<k$, transversality implies that $N$ and $\xi_s$ do not intersect. Furthermore, $\xi_s:[0,1]\to \rr^k$
is $C^{\infty}$.
Since $\Vert \xi_s(t) - \xi(t) \Vert \leq s < \epsilon$ and $\gamma(0)=\gamma(1)=0$, this concludes the proof.
\end{proof}

\begin{lem}\label{gp}
Let $m,k\in\nan$ and $N_1,\ldots,N_m$ be submanifolds of $\rr^k$ without boundary 
with $\dim N_i+1<k$, $i=1,\ldots,m$.
Suppose that $\xi:[0,1]\to \rr^k$ is a $C^{\infty}$-map with
$\xi(0),\xi(1)\in \cap_{i=1}^m N_i^c$.
Then, there exists a $C^{\infty}$-map $\zeta:[0,1]\to \rr^k$ such that
\begin{align}
&\zeta([0,1])\cap \lmk \bigcup_{i=1}^m N_i\rmk=\emptyset,\nonumber\\
&\xi(0)=\zeta(0),\; \xi(1)=\zeta(1).
\end{align}
\end{lem}
\begin{proof}
We consider the following statement for $j=1,\dots,m$.

(P$_{j}$):
There exists a $C^{\infty}$-map $\xi_{j}:[0,1]\to \rr^k$ such that
\begin{align}
&\xi_{j}([0,1])\cap \lmk \bigcup_{i=1}^j N_i\rmk=\emptyset,\nonumber\\
&\xi(0)=\xi_{j}(0),\; \xi(1)=\xi_{j}(1),\nonumber\\
&\sup_{t\in[0,1]}\lV \xi(t)-\xi_{j}(t)\rV_{\rr^k}<\lmk1-\frac 1{2^j}\rmk.
\end{align}

We prove (P$_{m}$) inductively, and set $\zeta=\xi_m$.
The statement (P$_{1}$) is obtained by applying Lemma \ref{deformation} to
$N_1$, $\xi$, and $\varepsilon=\frac 12>0$.
Assume that (P$_{j}$) is true for $j<m$. 
Applying  Lemma \ref{deformation} to $N_{j+1}$, $\xi_j$ and 
\[
\varepsilon_{j+1}=\min\left\{
2^{-j-1},\; \dr(\xi_j([0,1], \bigcup_{i=1}^j N_i))
\right\}>0,
\]
we obtain a $C^{\infty}$-path $\xi_{j+1}:[0,1]\to \rr^k$ satisfying
\begin{align}\label{jpo}
&\xi_{j+1}([0,1])\cap N_{j+1}=\emptyset,\nonumber\\
&\xi(0)=\xi_{j}(0)=\xi_{j+1}(0),\; \xi(1)=\xi_{j}(1)=\xi_{j+1}(1),\nonumber\\
&\sup_{t\in[0,1]}\lV \xi_{j+1}(t)-\xi_{j}(t)\rV_{\rr^k}<\varepsilon_{j+1}.
\end{align}
By the choice of $\varepsilon_{j+1}\le \dr(\xi_j([0,1], \bigcup_{i=1}^j N_i))$
and the last inequality in (\ref{jpo}), we conclude that  $\xi_{j+1}$ does not intersect with $\bigcup_{i=1}^j N_i$.
Again, by the choice of $\varepsilon_j\le 2^{-j-1}$ and the last inequality in (\ref{jpo})
we have $\sup_{t\in[0,1]}\lV \xi(t)-\xi_{j+1}(t)\rV_{\rr^k}<(1-\frac 1{2^{j+1}})$.
\end{proof}
\begin{lem}\label{skp}
Let $S_k$ be the subset in $\cc^k$ defined by (\ref{sk}).
Then for any points $\lambda,\mu\in S_k$, there exists 
 a $C^{\infty}$-path $\zeta:[0,1]\to \cc^k$ such that 
$\zeta([0,1])\subset S_k$ and
$\zeta(0) = \lambda$, $\zeta(1) = \mu$.
\end{lem}
\begin{proof}
We identify $\bbC^l$ with $\bbR^{2l}$ for any $l\in\nan$ naturally. For $j=1,\ldots,k$, and $(i,j)\in {\mathcal P}_k$, set
\begin{align*}
M_j:=\ltm v\in\cc^k\mid v_j=0\rtm,\quad  
Z_{(i,j)}:= \ltm v\in\cc^k\mid v_i=v_j\rtm.
\end{align*}
Furthermore, for any $(i,j), (l,m)\in {\mathcal P}_k$ with $(i,j)\neq(l,m)$, set
\begin{align*}
N_{(i,j),(l,m)}:=\ltm
v\in\cc^k\mid v_iv_m=v_jv_l\rtm\cap \lmk \cap_{j=1}^k M_j^c\rmk.
\end{align*}
Clearly $M_j$, $j=1,\ldots,k$ and $Z_{(i,j)}$,  $(i,j)\in {\mathcal P}_k$ are $2k-2$ dimensional submanifolds of $\rr^{2k}$ without boundary.
For any $(i,j), (l,m)\in {\mathcal P}_k$ with $(i,j)\neq(l,m)$, the map $\cc^k \ni v\to v_iv_m-v_jv_l\in \cc$ is a submersion on $ \cap_{j=1}^k M_j^c$.
Therefore, $N_{(i,j),(l,m)}$ is a $2k-2$ dimensional submanifold of $\rr^{2k}$ without boundary.
It is easy to see
\begin{align*}
S_k=\lmk \bigcap_{j=1}^k M_j^c\rmk \cap\lmk \bigcap_{(i,j)\in {\mathcal P}_k} Z_{(i,j)}^c\rmk \cap \lmk \bigcap _{(i,j), (l,m)\in {\mathcal P}_k,(i,j)\neq(l,m)}N_{(i,j),(l,m)}^c\rmk.
\end{align*}

Let $\xi:[0,1]\to \cc^k$ be a $C^{\infty}$-path defined by
\[
\xi(t)=(1-t)\lambda+t\mu, \qquad t\in[0,1],
\]connecting $\lambda$ and $\mu$.
Applying Lemma \ref{gp} to $\xi$ and the finite set of $2k-2$ dimensional submanifolds 
$M_j$, $Z_{(i,j)}$, $N_{(i,j),(l,m)}$
of $\rr^{2k}$ without boundary, we obtain the result. 
\end{proof}
\begin{lem}\label{pol}
Let $\lambda=(\lambda_1,\ldots,\lambda_k)\in\cc^k$,
and $n_1,\ldots,n_k\in\nan$ with $n_i\neq n_j$ if 
$(i,j)\in{\mathcal P}_k$. Let 
$\lambda(t)=(\lambda_1(t),\ldots,\lambda_k(t))\in\cc^k$ be defined by
\begin{align*}
\lambda_i(t):=\lmk \lambda_i+t^{2^{n_i}}
\rmk,
\end{align*}
for $t\in \rr$.
Then there exists $\delta>0$ such that $\lambda(t)\in S_k$ for all $t\in (0,\delta)$.
\end{lem}

\begin{proof}
First we show that for $n,m,l,k\in\nan$,
\[
2^n+2^m=2^l+2^k,
\]
implies
either $n=l$ and $m=k$,
or $n=k$ and $m=l$.
To see this, we may assume
$n\ge m$, $l\ge k$.
We have 
\[
2^{m}(2^{n-m}+1)=2^{k}(2^{l-k}+1).
\]
If $n\neq m$ and $k\neq l$, then 
$m=k$.
This implies $n-m=l-k$ that $n=l$.
If $n=m$, then 
we have $2^{m+1}=2^{k}(2^{l-k}+1)$.
Therefore, $l=k$ and $k=m=n$.
Similarly, if $l=k$, then $n=m=k=l$.

Let $\lambda_1,\ldots,\lambda_k\in\cc$,
and $n_1,\ldots,n_k\in\nan$ with $n_i\neq n_j$ if 
$(i,j)\in{\mathcal P}_k$.

For each $i=1,\ldots,k$, by the analiticity of
polynomial,
there exists a $\delta_i>0$
such that
\begin{align*}
\lambda_i+t^{2^{n_i}}\neq 0,
\end{align*}
for all $t\in (0,\delta_i)$. 
For any $(i,j)\in{\mathcal P}_k$,
there exists a $\delta_{ij}>0$ such that
\begin{align*}
\lambda_i+t^{2^{n_i}}\neq
\lambda_j+t^{2^{n_j}}
\end{align*}
for all $t\in (0,\delta_{ij})$, as 
$n_i\neq n_j$.
For any $(i,j),(i',j')\in{\mathcal P}_k$ with $(i,j)\neq(i',j')$,
from the above claim, we have 
\[
2^{n_i}+2^{n_{j'}}\neq 
2^{n_j}+2^{n_{i'}}.
\]
Therefore, the polynomial
\begin{align*}
\lmk
\lambda_i+t^{2^{n_i}}\rmk
\lmk
\lambda_{j'}+t^{2^{n_{j'}}}
\rmk-
\lmk
\lambda_j+t^{2^{n_j}}\rmk
\lmk
\lambda_{i'}+t^{2^{n_{i'}}}
\rmk
\end{align*}
is not zero.
Hence there exists a $\delta_{(ij),(i'j')}>0$
such that 
\begin{align*}
\lmk
\lambda_i+t^{2^{n_i}}\rmk
\lmk
\lambda_{j'}+t^{2^{n_{j'}}}
\rmk\neq
\lmk
\lambda_j+t^{2^{n_j}}\rmk
\lmk
\lambda_{i'}+t^{2^{n_{i'}}}
\rmk
\end{align*}
for 
$t\in(0,\delta_{(ij),(i'j')})$.
Hence setting $\delta:=\min\{\delta_i,\delta_{i,j},
\delta_{(i,j),(i',j')}\}$,
we have $\lambda(t)\in S_k$ for all $t\in (0,\delta)$.
\end{proof}

We close this section with a simple lemma.
\begin{lem}\label{sf}
Let $F$ be a finite subset of $\cc$ and let $\chi,\eta\in\cc$ with 
$\chi\neq\eta$. Then there exists a $C^{\infty}$-map $\xi:[0,1]\to \cc$
with $\xi(0)=\chi$, $\xi(1)=\eta$ and $\xi(t)\in F^c$
for all $t\in(0,1)$, such that
\[
\lv
\xi(t)-\chi
\rv\le 2\lv \eta-\chi \rv,\quad
t\in[0,1].
\]
\end{lem}
\begin{proof}
This can be done by a modification of the path $[0,1]\ni t\mapsto (1-t)\chi+t\zeta$ avoiding $F$.
\end{proof}
%
%
\section{Perturbation of Jordan matrices}

Here, we consider matrices $A\in\Mat_k(\bbC)$ that are close to a Jordan matrix in the sense that $A$ has the same block form as a Jordan matrix, but in each block, all diagonal elements are different. We exhibit explicitly 
the matrix diagonalizing it.
\begin{lem}\label{jd}
Let $k\in\nan$ and $n_1,\ldots,n_M\in\nan$
such that $n_1+\cdots+n_M=k$.
Decompose $\cc^k=\cc^{n_1}\oplus\cdots\oplus \cc^{n_M}$
and let $\{f_{\alpha}^{(l)}\}_{\alpha=1}^{n_l}$ 
be an orthonormal basis
of $\cc^{n_l}$, for each $l=1,\ldots,M$.
Let $\{\lambda_{\alpha}^{(l)}\}_{\alpha=1,\ldots, n_l,\;l=1,\ldots,M}$
be distinct elements in $\cc$. 
Define
\begin{align*}
J_l:=\sum_{\alpha=1}^{n_l}
\lambda_{\alpha}^{(l)}\ket{f_{\alpha}^{(l)}}
\bra{f_{\alpha}^{(l)}}
+\sum_{\alpha=2}^{n_l}
\ket{f_{\alpha-1}^{(l)}}\bra{f_{\alpha}^{(l)}},
\end{align*}
for $l=1,\ldots,M$, and let $J:=J_1\oplus\cdots\oplus J_M
$. Then there exists a diagonal matrix $D$ and an invertible matrix $P$ such that
\begin{align*}
J=PDP^{-1}.
\end{align*}
Here $D=D_1\oplus\cdots \oplus D_M$, 
$P=P_1\oplus\cdots\oplus P_M$ with respect to the decomposition
 $\cc^k=\cc^{n_1}\oplus\cdots\oplus \cc^{n_M}$, and
\begin{align*}
D_l=\sum_{\alpha=1}^{n_l}
\lambda_{\alpha}^{(l)}\ket{f_{\alpha}^{(l)}}
\bra{f_{\alpha}^{(l)}},\quad
P_l=\sum_{\alpha,\beta=1}^{n_l}
c_{\beta\alpha}^{(l)}\ket{f_{\beta}^{(l)}}\bra{f_{\alpha}^{(l)}},\quad
P_l^{-1}=\sum_{\alpha,\beta=1}^{n_l}
d_{\alpha\beta}^{(l)}\ket{f_{\alpha}^{(l)}}\bra{f_{\beta}^{(l)}}
\end{align*}
with
\begin{align*}
c_{\beta\alpha}^{(l)}
=\begin{cases}
\prod_{j=\beta}^{\alpha-1}\frac{1}{\lambda_{\alpha}^{(l)}- \lambda_j^{(l)}} &\alpha>\beta \\
1& \alpha=\beta\\
0& \alpha<\beta
\end{cases}
\,,\qquad
d_{\alpha\beta}^{(l)}
=
\begin{cases}
\prod_{j=\alpha+1}^{\beta}\frac{1}{\lambda_{\alpha}^{(l)}-\lambda_j^{(l)}}& \alpha<\beta \\
1& \alpha=\beta\\
0& \alpha>\beta
\end{cases}\,.
\end{align*}
\end{lem}
\begin{proof}
This follows by the matrix diagonalizations of the blocks $J_l$.
\end{proof}
%
%
\section{Primitive maps}
In this section we collect known results about positive maps on $\Mat_k(\cc)$. We refer the reader to the literature for proofs of the stated theorems, as for example to the notes~\cite{Wolf:2012aa}, and references therein.
\begin{thm}\label{pm1}
Let $T:\Mat_k(\cc)\to \Mat_k(\cc)$ be a
positive linear map. The following properties are equivalent:
\begin{enumerate}
\item There is no nontrivial 
orthogonal projection $P$ such that
$T(P\Mat_k(\cc)P)\subset
P\Mat_k(\cc)P$,
\item For any
nonzero $A\geq0$ and $t>0$,
$\exp(tT)(A)>0$
\end{enumerate}
\end{thm}
\begin{rem}
A positive map satisfying the above equivalent conditions is said to be irreducible.
\end{rem}
We say that $\lambda$ is a nondegenerate eigenvalue of $T$
if the corresponding projection $P_{\{\lambda\}}^T$ is one dimensional. Irreducible positive maps satisfy the following properties.
\begin{thm}\label{pm2}
Let
$T:\Mat_k(\cc)\to \Mat_k(\cc)$ be a nonzero irreducible positive linear 
map.
Then the spectral radius $r_T$ of $T$ is a 
strictly positive, non-degenerate eigenvalue with a strictly positive 
eigenvector $h_T$:
\begin{equation*}
T(h_T) = r_T h_T >0.
\end{equation*}
\end{thm}
\begin{thm}\label{pm3}
Let $T:\Mat_k(\cc)\to \Mat_k(\cc)$ be a unital completely positive map and let
\begin{align*}
T(A)=\sum_{i=1}^{n}B_iAB_i^* 
\end{align*}
be its Kraus decomposition. Let $\bbB := (B_1,\ldots,B_n)$. Then the following properties are equivalent:
\begin{enumerate}
\item
There exists $l\in\nan$ such that $T^l(A)>0$
for any nonzero $A\geq 0$,
\item
There exists a unique faithful $T$-invariant state $\varphi$, and it satisfies
 \[
\lim_{l\to\infty} T^l(A)=\varphi(A) 1,\quad
A\in \Mat_k(\cc),
\]
\item
$\sigma(T)\cap\{z\in\cc:|z|\ge 1\}=\{1\}$, $1$ is a nondegenerate eigenvalue of $T$, and there exists a faithful $T$-invariant state,
\item
There exists $m\in\nan$
such that 
${\mathcal K}_m(\bb)=\Mat_k(\cc)$, where ${\mathcal K}_m(\bb)$ was defined in~(\ref{KB}),
\item
There exists $m\in\nan$
such that 
${\mathcal K}_l(\bb)=\Mat_k(\cc)$, for all $l\ge m$.
\end{enumerate}\end{thm}
\begin{rem}\label{Def:Primitive}
A unital completely positive map satisfying the above
(equivalent)
conditions is said to be primitive.
\end{rem}

\section{Continuous paths of positive maps}\label{sec:PathMaps}
We gather simple results for spectral quantities of elements in $\caT_k$ or $\caT_k$, see Section~\ref{sec:result}, and on continuous paths of such maps. 
\begin{lem}\label{lem : left and right eigenvectors}
For any $T\in{\caT_k}$,
\begin{equation*}
e_T=P_{\{r_T\}}^T(1)
\end{equation*}
is a nonzero element in $\mk_+$. Moreover, there exists a unique $r_T^{-1}T$-invariant state $\varphi_T$ on $\mk$ such that for any $a\in \mk$, 
\begin{equation*}
P_{\{r_T\}}^T(a)=\varphi_T(a)e_T,
\end{equation*}
and $\varphi_T(e_T)=1$.
\end{lem}
\proof
For $T\in\caT_k$, note that $Te_T = r_T e_T$.
By the definition of $\caT_k$,  we see that
\begin{equation}\label{PT}
P^{T}_{\{r_T\}}(a)=\lim_{n\to\infty} r_T^{-n}T^n(a),
\end{equation}
for all $a\in\mk$. 
From this and the positivity of $T$, 
$e_T=P^{T}_{\{r_T\}}(1)\ge 0$.
Again by the positivity of $T$ and (\ref{PT}), $e_T=0$ would imply $P^{T}_{\{r_T\}}=0$,
which is not true. Therefore, $e_T$ is nonzero.
 Furthermore, as the range $P^{T}_{\{r_T\}}$ is the one dimensional ray of $e_T$, there exists a linear functional $\varphi_T$ such that
\begin{equation}\label{vp}
P^{T}_{\{r_T\}}(a)=\varphi_T(a)e_T,
\end{equation}
for all $a\in\mk$. Finally, from (\ref{PT}), $\varphi_T$ is a $r_T^{-1}T$-invariant state.
As $\varphi_T(e_T) e_T=P^{T}_{\{r_T\}}(e_T)=e_T$, we have $\varphi_T(e_T)=1$.
\endproof

\begin{lem}\label{lem: continuity and decay}
For a continuous  and piecewise $C^1$-path $T:[0,1]\to{\caT_k}$,
the corresponding paths
$e_{T_t}$, $\varphi_{T_t}$, and $r_{T_t}$ are continuous and piecewise $C^1$. Moreover, there exist $0<\lambda<1$ and $c>0$ such that
\begin{equation*}
\sup_{t\in[0,1]}\lV r_{T_t}^{-l}T_t^l(1-P_{\{r_{T_t}\}}^{T_t})\rV \le c \lambda^l,
\end{equation*}
for all $l\in\bbN$.
\end{lem}
\proof
For each $t_0\in [0,1]$, $T_{t_0}\in\caT_k$.
Therefore, there exists a ${\delta_{t_0}}$ with $\frac 13 r_{T_{t_0}}>{\delta_{t_0}}>0$ such that
$\sigma\lmk T_{t_0}\setminus\{ r_{T_{t_0}}\}\rmk \subset B_ {(r_{T_{t_0}}-3{\delta_{t_0}})}(0)$.
Fix such a ${\delta_{t_0}}$.
Then we have
\begin{align}
P^{T_{t_0}}_{\{r_{T_{t_0}}\}}
=\frac{1}{2\pi i}\oint_{|z-r_{T_{t_0}}|={\delta_{t_0}}}
(z- T_{t_0})^{-1}dz.
\end{align}
 By the continuity of $T$,
there exists an ${\varepsilon_{t_0}}>0$
such that
\begin{align}\label{sd}
\sigma(T_t)\subset (\sigma(T_{t_0}))_{\frac {{\delta_{t_0}}}{2}}
\end{align}
for any $t\in (t_0-{\varepsilon_{t_0}},t_0+{\varepsilon_{t_0}} )\cap [0,1]$.
For this  ${\varepsilon_{t_0}}>0$ fixed, 
\[
(t_0-{\varepsilon_{t_0}},t_0+{\varepsilon_{t_0}} )\cap [0,1]\ni t\mapsto
 Q_t:=\frac{1}{2\pi i}\oint_{|z-r_{T_{t_0}}|={\delta_{t_0}}}
(z- T_{t})^{-1}dz
\]
is well-defined, continuous and piecewise $C^1$.
From the continuity of the path of projections  $Q_t$, and  the fact that
$P^{T_{t_0}}_{\{r_{T_{t_0}}\}}=Q_{t_0}$ is one dimensional, we see that each $Q_t$ is a one dimensional projection.
Hence $Q_t$ is a one dimensional projection corresponding to the spectrum
 $\sigma(T_t)\cap B_{{\delta_{t_0}}}(r_{T_{t_0}})$. From this,
(\ref{sd}) and the fact $r_{T_t}\in\sigma(T_t)$, we get 
\begin{align}\label{ts}
\{r_{T_t}\}=\sigma(T_t)\cap B_{{\delta_{t_0}}}(r_{T_{t_0}}).
\end{align}
 Therefore we obtain
 \[
 Q_t=P^{T_{t}}_{\{r_{T_{t}}\}},\qquad t\in (t_0-{\varepsilon_{t_0}},t_0+{\varepsilon_{t_0}} )\cap [0,1].
 \]
 In particular,
 $[0,1]\ni t\mapsto P^{T_{t}}_{\{r_{T_{t}}\}}$ is continuous and piecewise $C^1$.
 
 From this and the following formulae,
\begin{align}
e_{T_t}=P^{T_{t}}_{\{r_{T_{t}}\}}(1),\qquad
\varphi_{T_t}=
\frac{\Tr_{\mk} \lmk P^{T_{t}}_{\{r_{T_{t}}\}}(\cdot)\rmk }
{\Tr_{\mk}(e_{T_t})},
\qquad
 r_{T_t}=
\frac{\Tr_{\mk}\lmk T_t(e_{T_t})\rmk}{\Tr_{\mk}(e_{T_t})},
\end{align}
we see that $e_{T_t}$, $\varphi_{T_t}$, and $r_{T_t}$ are continuous and piecewise $C^1$.
 
To prove the second part, for each $t_0\in [0,1]$, we take ${\delta_{t_0}}>0$ and ${\varepsilon_{t_0}}>0$ as above.
From (\ref{sd}) and (\ref{ts}), we have
 \[
 \sigma(T_t)\subset B_ {(r_{T_{t_0}}-\frac{5}{2}{\delta_{t_0}})}(0)\cup \{r_{T_t}\},\qquad t\in (t_0-{\varepsilon_{t_0}},t_0+{\varepsilon_{t_0}} )\cap [0,1].
 \]
 By the continuity of $r_{T_t}$, this implies the existence of $0<\varepsilon'_{t_0}<\varepsilon_{t_0}$ and $0<\delta'_{t_0}<1$ such that
\[
\sigma(r_{T_t}^{-1}T_t)\setminus \{1\}\subset
B_{1-\delta'_{t_0}}(0),\qquad
 t\in (t_0-{\varepsilon'_{t_0}},t_0+{\varepsilon'_{t_0}} )\cap [0,1].
\]
By the compactness of $[0,1]$, there exist finite number of $t_1,\ldots,t_m\in[0,1]$ such that
\[
[0,1]=\cup_{i=1}^m(t_i-{\varepsilon'_{t_i}},t_i+{\varepsilon'_{t_i}} )\cap [0,1].
\]
Set $\delta:=\min\{\delta'_{i}\mid i=1,\ldots,m\}>0$.
Then we have
 \[
\sigma(r_{T_t}^{-1}T_t)\setminus \{1\}\subset
B_{1-\delta}(0),\qquad
 t\in [0,1].
\]
Setting
\begin{equation*}
c:=\sup_{(z,t):|z|=1-\delta ,t\in[0,1]}
\lV (z-r_{T_t}^{-1}T_t)^{-1}\rV<\infty,
\end{equation*}
we obtain
\begin{align*}\lV
r_{T_t}^{-l}T_t^l(1-P_{\{r_{T_t}\}}^{T_t})\rV
=
\lV
\frac 1 {2\pi i}\oint_{|z|=1-\delta} z^l
(z-r_{T_t}^{-1}T_t)^{-1} dz
\rV
\le c (1-\delta)^l,\qquad
l\in\nan,
\end{align*}
and the claim follows with $\lambda = 1-\delta$.
\endproof

By definition, the two quantities
\begin{equation*}
a_{T}=\lV (pe_Tp)^{-1}\rV_{\mk},\qquad c_{T}=\lV (q\rho_Tq)^{-1}\rV_{\mk},
\end{equation*}
are finite.
\begin{lem}\label{lem: Inverses}
Let $k\in\nan$ and $p,q$ be two fixed orthogonal projections on $\mk$. For any continuous path $T:[0,1]\to \caT_k$, 
$a_{T_t,p}$ and $c_{T_t,q}$ are continuous. In particular, $\sup_{t\in[0,1]}a_{T_t,p}<\infty$, and $\sup_{t\in[0,1]}c_{T_t,q}<\infty$.
\end{lem}
\begin{proof}
By the proof of the previous lemma and since $\caT_k\subset\caT_k$, $e_{T_t}$ and $\varphi_{T_t}$ are continuous and so is $\rho_{T_t}$. 
Hence, $t\mapsto a_{T_t,p}$ and $t\mapsto c_{T_t,q}$ are continuous as well, and since they are defined on a compact set, they are uniformly bounded.
\end{proof}
\section{Path of vector spaces}
The proof of the following lemma is standard.
\begin{lem}\label{poa}
Let $k,m\in\nan$ with $k\le m$.
Let $
X:[0,1]\to (\Mat_m(\cc))_+
$
be continuous and piecewise $C^1$-path of positive matrices such that
the rank of $X(t)$ is $k$ for all $t\in[0,1]$.
Let $S(t)$ be the support projection of $X(t)$,
and set $\gamma(t):=\dc (\sigma(X(t))\setminus\{0\},\{0\})$.
 Then, the path of projections
 \[
 S:[0,1]\ni t\mapsto S(t)\in \Mat_m(\cc)
 \]
 is continuous and piecewise $C^1$
 and 
 \[
 \inf_{t\in[0,1]}\gamma(t)>0.
 \]
\end{lem}
\begin{lem}\label{pob}
Let $l,k,m\in\nan$ with $k\le m$.
Let $\psi_i:[0,1]\to\cc^m$, $i=1,\ldots, l$ be continuous and piecewise 
$C^1$-paths of vectors in $\cc^m$, such that
\[
\dim\spn \{\psi_i(t)\}_{i=1}^l=k,\qquad t\in[0,1].
\]
 For each $t\in[0,1]$, 
 let $S(t)$ be orthogonal projection onto the span of $\{\psi_i(t)\}_{i=1}^l$.
 Then, the path of projections
 \[
 S:[0,1]\ni t\mapsto S(t)\in \Mat_m(\cc)
 \]
 is continuous and piecewise $C^1$.
\end{lem}
\proof
Define 
\[
X(t):=\sum_{i=1}^l \ket{\psi_i(t)}\bra{\psi_i(t)}.
\]
Then $X:[0,1]\ni t\mapsto X(t)\in \Mat_m(\cc)_+$ defines a continuous and piecewise $C^1$-path, and $S(t)$ is the support projection
of $X(t)$. Hence the rank of $X(t)$ is $k$ for all $t\in[0,1]$.
Applying Lemma \ref{poa}, we obtain the claim.

\endproof

\end{document}